\documentclass[12pt]{amsart}

\usepackage{amsmath,amsthm}
\usepackage{amssymb}
\usepackage[all]{xy}
\usepackage{mathrsfs}

\usepackage{color}
\usepackage{hyperref}
\usepackage{graphicx}

\renewcommand{\S}{{\mathbb S}}

\newcommand{\R}{{\mathbb R}}
\newcommand{\D}{{\mathcal D}}
\newcommand{\C}{{\mathbb C}}

\def\tilde{\widetilde}
\def \bfo {\begin {eqnarray*} }
\def \efo {\end {eqnarray*} }
\def \ba {\begin {eqnarray*} }
\def \ea {\end {eqnarray*} }
\def \beq {\begin {eqnarray}}
\def \eeq {\end {eqnarray}}
\def \supp {\hbox{supp}\,}

\def\diag{\hbox{diag }}

\def \p {\partial}

\def\F{{\mathcal F}}

\newtheorem{definition}{Definition}[section]
\newtheorem{theorem}[definition]{Theorem}
\newtheorem{lemma}[definition]{Lemma}
\newtheorem{problem}[definition]{Problem}
\newtheorem{proposition}[definition]{Proposition}
\newtheorem{corollary}[definition]{Corollary}

\DeclareMathOperator{\WF}{WF}

\title[Inverse  problem of finding cosmic strings]
{On the inverse  problem of finding cosmic strings and other topological defects}
\author[Lassas]{Matti Lassas}
\address{Department of Mathematics and Statistics, University of Helsinki, Box 68, Helsinki, 00014, Finland}
\email{Matti.Lassas@helsinki.fi}
\author[Oksanen]{Lauri Oksanen}
\address{Department of Mathematics, University College London, Gower Street, London, WC1E 6BT, UK}
\email{l.oksanen@ucl.ac.uk}
\author[Stefanov]{Plamen Stefanov}
\address{Department of Mathematics, Purdue University, West Lafayette, IN 47907, USA}
\email{stefanop@purdue.edu}
\author[Uhlmann]{Gunther Uhlmann}
\address{Department of Mathematics, University of Washington, 
Box 354350
Seattle, Washington 98195,
USA\footnote{Also affiliated to Department of Mathematics and Statistics, University of Helsinki, Finland,
and 
HKUST Jockey Club Institute for Advanced Study, HKUST, Clear Water Bay, Kowloon, Hong Kong, China.
}}
\email{gunther@math.washington.edu}
\date{May 11, 2015}

\keywords{Inverse problems, tomography, Cosmic Microwave Background, 
cosmic strings, topological defects}

\begin{document}

\begin{abstract}
We consider how  microlocal methods developed for tomographic problems
can be used to detect singularities of the Lorentzian metric of the Universe
using measurements of the Cosmic Microwave Background radiation.
The physical model we study is mathematically rigorous but highly idealized.
\end{abstract}

\maketitle

\section{Introduction}\label{introduction 1}

We study the dectection of singularities 
of the Lorenzian metric of the Universe
from Cosmic Microwave Background (CMB) radiation measurements.
The singularities are considered in the sense of the wave front set that describes where the metric is non-smooth in the spacetime and
also in which direction the singularity occurs. The direction of the singularity is characterized by using the Fourier transform of the metric, see Definition \ref{def_WF} below.

%(caused by singularities of the stress-energy tensor)

A singularity in the metric could be caused for example by a cosmic string \cite{Anderson1,Vil1}.
A cosmic string is a singularity in the stress energy tensor that is supported on a two-dimensional timelike surface in the 
spacetime. 
% that could have formed in the early Universe during
%the phase transitions. 
%Among the possible topological defects of the Universe, cosmic strings
%have raised particular interest as 
The existence of cosmic strings finds support in super-string theories 
\cite{Sar},
however, 
there is no direct connection between string theory and the theory of cosmic strings.
We refer to \cite{PlanckXIX,PlanckXXV,Sazhina} regarding the existence (or inexistence) of cosmic strings in view of CMB measurements collected by the Planck Surveyor mission in 2013. 

%We will consider the detection of general type of singularities of the Lorentzian metric.
% $g$ and the related singularities of the stress-energy tensor $T$ modeling the mass and electromagnetic energy
%distribution in the space-time. 
The singularities of which potential detectability is interesting to study
include cosmic stings, monopoles, cosmic walls and black holes.
There is a vast physical literature concerning the 
effects of particular types of singularities or topological defects on the CMB measurements, see e.g. \cite{Coulson,Cruz,Fraisse} and references therein.
The contribution of the present paper is to adapt techniques from the mathematical study of inverse problems to CMB measurements. These techniques allow us to detect singularities without apriori knowledge of their geometry.
Hence it might be also possible to detect singularities that are not predicted by the current physical knowledge. 
Furthermore, the techniques allow us to study the opposite question, that is, what type of singularities are invisible to our measurements and therefore can not be detected \cite{GKLUbul, Greenleaf2003a}.

%{\em
%Regarding cosmic strings,

%Cosmic strings with large density produce duplicate images or various kind of discontinuities in fluctuations in the CMB data due to the Sachs-Wolfe effect \cite{SWo}, and 

%it was predicted that their Sachs-Wolfe effect \cite{SWo} could be detectable in the CMB data collected by the Planck Surveyor mission
%\cite{Fraisse}. However,
%the 2013 analysis of data from the Planck mission did not to find any direct evidence of cosmic strings \cite{PlanckXIX,PlanckXXV}. 
%This does not rule out the existence of cosmic strings having  
%string tension less than  

%G?/c2²4.83?10?7 

%$G\mu/c^2\leq 4.83\,\cdotp 10^{-7}$,
%and a wavelet analysis of the data by Sazhina et al \cite{Sazhina} indicates that there are 
% several  string candidates. 

%We emphasize that the important feature of the method
%developed in the present paper is that we do not need to know a priori what type of singularities we are trying to detect from the CMB data.
%}

Several types of measurements have been proposed in astrophysical literature for 
detection of topological defects of Universe.
These
include optical measurements, such as gravitational lensing effects caused by cosmic stings
\cite{Agol,Sazhin,Schild},
observations of density of mass in Universe, measurements of gravitational waves
\cite{Aasi}, and
the temperature changes in the Cosmic Microwave Background \cite{PlanckXIX,PlanckXXV,Fraisse, Sazhina}.
We show that the CMB measurements have a tomographic nature, 
and concentrate
on them in this paper.

The detection of singularities has been
extensively studied in microlocal theory of tomography, 
see e.g. \cite{Greenleaf1989, Greenleaf1990, Quinto1993} and the review \cite{FLU}.
In many problems related to tomography, for instance in medical imaging,
it has been shown that measurements can be used to
detect singularities. Moreover, the visible singularities have been characterised
in many cases. 
In this microlocal context, a singularity is considered to be invisible if it causes only a smooth perturbation in the measured data. 

%As topological defects, such as the cosmic strings, correspond to 
%singularities of the stress-energy tensor and the metric tensor,
%it is natural to ask if microlocal methods
%can explain which type of topological defects can be in principle detected from the CMB measurements.
In the present paper we characterize the singularities of the Lorenzian metric that are visible from a linearization of the CMB measurements, and that move slower than the speed of light.
We show also that all singularities moving faster than the speed of light are invisible.
We do not analyze recovery of singularities moving at the speed of light.
Our approach is based on a highly idealized deterministic model of CMB measurements,
but such a model can be viewed as a first step in developing tomographic methods for
more realistic, possible stochastic, models of
the CMB measurements, see e.g. \cite{Ha}. 

We obtain the characterization of visible singularities via microlocal analysis of the geodesic ray transform on a
Friedmann-Lema\^{i}tre-Robertson-Walker type spacetime. The transform is restricted on light rays and we call it the light ray transform. 
Such a light ray transform has been studied in $1+2$ dimensions in \cite{Guillemin1989}, and the $1+3$ dimensional light ray transform, as considered in the present paper, belongs to the class of Fourier integral operators studied in \cite{Greenleaf1991}.

Contrary to \cite{Greenleaf1991}, we avoid using the $I^{p,l}$ calculus
by first microlocalizing on spacelike covectors. Then
we invert microlocally the light ray transform %acting on symmetric 2-tensors 
up to potential fields, see (\ref{def_potential_field}) below for the definition, and conformal multiples of the metric of the spacetime. 
This is sharp since those two subspaces belong to the kernel of the linearization and they correspond to the gauge invariance of the non-linear problem under diffemorphisms and conformal changes. 
The spacelike covectors correspond physically to singularities moving slower than the speed of light.

This paper is organized as follows. Section 1 is the introduction and Section 2 introduces some notations. In Section 3 we formulate the inverse problem for the CMB measurements and its linearization, and in Section 4 we state the main results. Section 5 deals with the parametrization of the CMB measurements. In Section 6 we describe the conformal invariance inherent in the problem. 
Section 7 contains the reduction of the linearized problem to inversion of the light ray transform. 
In Section 8 we study the light ray trasform in a translation invariant case and express its normal operator as a Fourier multiplier. 
This motivates our subsequent study in the general case which is not translation invariant due to the fact that measurements are available only in a small set in the spacetime. 
In Section 9 we study the null space of the light ray transform that corresponds to the gauge invariance of the problem.  
In Sections 10 and 11 we characterize the visible spacelike singularities. Moreover, we compute the symbol of the normal operator of the light ray transform on the cone of spacelike covectors, on which it is a pseudodifferential operator. 

\section{Notations}
\def\Sym{\mathop{Sym^2}}
\def\L1{\Lambda^1}
\def\E{\mathcal E}

Let $M \subset \R^n$ be open.
We denote by $\D'(M)$ the distributions on $M$
and by $\E'(M)$ the distributions with compact support.
Moreover, we denote by $C^m(M)$, $m = 0,1,2,\dots, \infty$, the space of $m$-times continuously differentiable functions, and let $C_0^m(M) \subset C^m(M)$ be the subspace of functions with compact support. 
For $m < \infty$, we define also the Banach space
$$
C_b^m(M) = \{u \in C^m(M);\ \sum_{|\alpha|\le m} \sup_{x \in M} |\p_x^\alpha u(x)| < \infty \}.
$$
We denote by $\Sym$ the symmetric 2-tensors on $M$,
and by $\L1$ the 1-forms on $M$.
We use the notations $\mathcal F(M; E)$,
$\mathcal F = \D', \E', C^m, C_b^m$; $E = \Sym, \L1, \C$,
for the subspaces $\F$ of distributions taking values on
the vector bundle $E$. 

We denote by $S^m(T^*M) = S_{1,0}^m(T^*M)$, $m \in \R$, the symbols as defined in \cite{FIO1}.
If $\chi \in S^m(T^*M)$ then we denote the corresponding pseudodifferential operator also by $\chi$.
We denote the wave front set of a distribution $u \in \D'(M;E)$
by $\WF(u)$. For reader's convenience we recall the definition here. 

\begin{definition}
\label{def_WF}
Let $u \in \D'(M;E)$, $E = \Sym, \L1, \C$.
A point $(x, \xi) \in T^*M \setminus 0$ is in the wave front set of $u$ 
if there does not exists an open cone $\Gamma \subset T_x^* M$ containing $\xi$ and 
a function $\phi \in C_0^\infty(M)$ satisfying $\phi(x) \ne 0$
such that for all $N = 1,2,\dots$ there is $C > 0$ such that the Fourier transform of $\phi u$ decays as follows:
$$
|\widehat{\phi u}(\xi)| \le C (1 + |\xi|)^{-N}, \quad \xi \in \Gamma.
$$
\end{definition}

Let us recall a typical example, see e.g. \cite[Th. 8.1.5]{HormanderVol1}. 
Let $u \in C_0^\infty(M)$, $k \le n$, and let us consider coordinates
$(x', x'') \in \R^{k \times (n-k)}$ on $M$. 
Then the Euclidean surface measure $\delta(x'')$ of $\{(x', 0); x' \in \R^k \}$ satisfies
$$
\WF(u \delta(x'')) = \{(x', 0; 0, \xi'') \in T^*M \setminus 0;\ (x', 0) \in \supp(u) \},
$$
where $\xi''$ is the covector corresponding to $x''$.
As all smooth surfaces in $M$ are locally of this form, it can be seen that the wave front set
of a surface measure
is the conormal bundle of the corresponding surface. 
This applies in particular to cosmic strings that are singularities supported on two dimensional surfaces, see e.g. \cite[Chapter 6]{Anderson1}. 

%Let $g$ be a Lorentzian metric on $M$ and let $x \in M$.
%We say that a vector $V \in T_x M$ is spacelike if $(V,V)_g > 0$,
%timelike if $(V,V)_g < 0$, and lightlike if $(V,V)_g = 0$.
%We call the lightlike vectors also the null vectors. 
%We use the same language also for covectors, that is,
%$\xi \in T_x^* M$ is spacelike if $(\xi,\xi)_g > 0$.

\section{Mathematical formulation of the CMB measurements}

\subsection{The data and the inverse problem}

Let us begin by formulating a mathematical problem on a time-oriented Lorentzian manifold $(M,g)$ of dimension $1+n$. 
Below, we will consider a linearized version of this problem, and solve the linearized problem only in a microlocal sense.  
We assume that $M$ is a smooth manifold
and that $g$ is a $C^2$-smooth metric tensor with signature $(-, +, \dots, +)$.

We recall that a point $(x,V)$ on the tangent bundle $T M$
is called spacelike if the inner product $(V,V)_g$
is strictly positive, timelike if $(V,V)_g < 0$, and lightlike if $(V,V)_g = 0$. A submanifold $\Sigma \subset M$ is called spacelike if
all the tangent vectors in $T\Sigma$ are spacelike,
and a geodesic $\gamma : [0, \ell] \to M$ is called a null geodesic if its 
tangent vector $\dot \gamma(\tau)$ is lightlike for one, and hence for all $\tau \in [0,\ell]$. 
Let us also recall that $(x,Z) \in TM$ is called an observer 
if $Z$ is future pointing and $(Z,Z)_{g} = -1$.

\def\UU{\mathcal U}
Let $\Sigma \subset M$ be a smooth spacelike submanifold of codimension one,
and let $\nu$ be the future pointing normal vector field on $\Sigma$
that is of unit length in the sense that $(\nu, \nu)_g = -1$.
Let $\UU$ be another smooth submanifold of $M$
and let $Z$ be a smooth section of $TM$ defined on $\UU$,
and suppose that $(x,Z)$ is an observer for all $x \in \UU$.

Let $E_0 > 0$ and let us consider a null geodesic $\beta$ satisfying
\begin{align}
\label{def_beta}
\beta(0) \in \Sigma, \quad (\dot \beta(0), \nu)_g = -E_0,
\end{align}
and suppose that $\beta(\tau) \in \UU$ for some $\tau \in \R$.
We define
\begin{align}
\label{def_Newtonian_velocity}
x = \beta(\tau), \quad
V = E^{-1}\dot \beta(\tau) -  Z,
\end{align}
where $E = -(\dot \beta(\tau), Z)_g$.
Note that $V$ satisfies
\begin{align}
\label{def_cel_spere}
(V, Z)_{g} = 0,\quad (V, V)_{g} = 1,
\end{align}
and that $(x,V)$
determines uniquely a null geodesic $\beta$
satisfying (\ref{def_beta}).
Indeed, $\beta(\tau) = x$ and 
$\dot \beta(\tau) = E(V+Z)$
determine a null geodesic up to an affine reparametrization,
that is, up to a choice of $\tau$ and $E$,
and (\ref{def_beta}) fixes the parametrization.

We denote by $\mathbb S_{g,x,Z}$
the set of vectors $V \in T_x M$
satisfying (\ref{def_cel_spere}).
The set $\mathbb S_{g,x,Z}$ is called the celestial sphere of the observer $(x,Z)$, see e.g. \cite{SW}. 
Furthermore, we denote by $\mathcal S_{g,Z} \UU$
the set of points $(x, V) \in TM$
such that $x \in \UU$, $V \in \mathbb S_{g,x,Z}$,
and that there is a null geodesic $\beta$ 
satisfying (\ref{def_beta})
and (\ref{def_Newtonian_velocity}) for some $\tau \in \R$.

Physically, the null geodesics $\beta$ satisfying 
(\ref{def_beta}) 
correspond to photons that 
are emitted with a fixed energy $E_0$ uniformly in all future pointing lightlike directions
on $\Sigma$.
Moreover, $E$ and $V$ in (\ref{def_Newtonian_velocity}) are the energy and the Newtonian velocity of $\beta$ 
as measured by the observer $(x, Z)$.
The proportional difference between the emitted and observed energies,
\def\dom{\text{Dom}}
\begin{align}
\label{def_redshift}
\mathcal R_{g,Z}(x, V) = \frac{E_0 - E}{E} = 
\frac{(\dot \beta(0), \nu)_{g}}{(\dot \beta(\tau), Z)_{g}}
- 1, \quad (x,V) \in \mathcal S_{g,Z} \UU,
\end{align}
is called the redshift of $\beta$ as measured by $(x, Z)$.
Here $\beta$ is the null geodesic satisfying (\ref{def_beta})
and (\ref{def_Newtonian_velocity}) for some $\tau \in \R$.
A general formulation of the inverse problem that we consider is the following:
\begin{problem}[Inverse problem for redshift measurements]
\label{IP}
Given the function $\mathcal R_{g,Z} : \mathcal S_{g,Z} \UU \to \R$
determine $(W,g)$
where $W\subset M$ is the union of the null
geodesics connecting points of $\mathcal U$ to points of $\Sigma$.
\end{problem}

The red shift measurements (\ref{def_redshift}) do not change under the 
group of transformations generated by diffeomorphisms and conformal changes. 
Indeed, if $\psi$ is a diffeomorpism that is identity on $\Sigma$ and $\UU$,
then the redshift data corresponding to $g$ and its pullback $\psi^*g$ are the same.
Next, if $c$ is a smooth strictly positive function on $M$ 
satisfying $c=1$ on $\Sigma$ and $\UU$, then the redshift data corresponding to $cg$ and $g$ are the same.
This is based on the fact that the null geodesics of $c g$ are the same as those of $g$ as point sets but they are parameterized differently,
see Lemma \ref{lem_confinv_nullgeod} below. 
Hence the determination of $(M,g)$ in Problem \ref{IP} should be understood modulo the gauge invariance given by this group of transformations.

\subsection{The linearized CMB inverse problem}

Let us now turn to a linearization of Problem \ref{IP} at a 
Friedmann-Lema\^{i}tre-Robertson-Walker type model,
%
%Let us consider a spacetime $(M,g)$, that is, a Lorentzian manifold
%satisfying the Einstein field equation
%\ba
%\hbox{Ric}_{jk}(g)-\frac 12 S(g)\,g_{jk}=T_{jk},
%\ea
%where Ric is the Ricci curvature tensor of $g$, $S$ is the scalar curvature
%of $g$ and $T$ is a stress-energy tensor. 
%The cosmic string would correspond
%to a stess-energy tensor $T$ with a delta-distribution type singularity supported 
%on a two dimensional submanifold of the space-time $M$,
%and in this case the metric tensor $g$ would be $C^{0,1}$-smooth.
%In this paper we focus on recovery of singularities of $g$.
%\HOX{Check or remove the claim on the cosmic string and on the WF!!
%We don't consider $T$ below. Also, add comment on the smoothness we need.}
%
%If we can find singularities of the metric tensor $g$,
%then we get information on the singularities of 
%the stess-energy tensor $T$
%by substituting $g$ to the Einstein
%field equation.
%Indeed, if $(x,\xi) \in \hbox{WF}(g) \subset T M$
%and $\xi$ is space-like or time-like,
%then $(x,\xi)\in\hbox{WF}(T)$. 
%
that is, at a Lorenzian manifold of the following warped product form
\begin{align}
\label{background_metric}
g(x) = -dt^2 + a^2(t) dy^2, \quad x =(t,y) \in M = (0, \infty) \times \R^3,
\end{align}
where $a > 0$ is smooth on $(0,\infty)$ and may have singularity as $t \to 0$.
A singularity at the boundary $t = 0$ 
corresponds physically to the Big Bang.
In particular, the choice $a(t) = t^{2/3}$ gives the Einstein-de Sitter cosmological model \cite[p. 31]{SW}. 
Occasionally we write $a(x) = a(t)$ for $x = (t,y)$.

Let $t_0 > 0$ be small and consider the surface
\begin{align}
\label{def_sigma}
\Sigma=\{ (t,y) \in M:\ t=t_0\}.
\end{align}
Then the photons satisfying (\ref{def_beta}) give a model for the CMB radiation. 
Physically, the time $t_0$ corresponds to the time when the Universe was cool enough
so that stable atoms could form. As these atoms could no longer absorb the thermal radiation, the photons that were around stayed forming the CMB radiation.

The model is highly idealized as the initial frequency spectrum of microwave 
background radiation would more likely follow e.g.\ Planck photon distribution
function, see e.g. \cite[Ex. 5.5.4]{SW}. 
A physically more realistic
model would allow the photons emitted at the surface $\Sigma$
to have an energy distribution with median $E_0$. 
This would lead to similar 
considerations as below, by assuming that we measure
the energy distribution of photons and find the median of this distribution.
However, for simplicity we consider the case when all emitted photons
have constant energy.

Let $g_\epsilon \in C^2(M;\Sym)$, $\epsilon \in [0,1]$, be a one parameter family of Lorentzian metrics,
and suppose that $g_\epsilon = g$ in $M \setminus M_0$ where
$$M_0 = (t_0, \infty) \times \R^3.$$
Here $t_0$ is as in (\ref{def_sigma}).
We will now describe the CMB measurements on $(M,g_\epsilon)$.
%Let $\beta_\epsilon$ be a future pointing null geodesic with respect to the metric $g_\epsilon$ satisfying(\ref{def_beta}).
%We call such geodesics CMB photons with respect to $g_\epsilon$.
Let $t_1 > t_0$, let $\UU_1 \subset \R^3$ be open and bounded,
and define $\UU = \{t_1\} \times \UU_1$.
To avoid technicalities in the exposition, we assume that 
$g_\epsilon = g$ on $\UU$
and consider the observers $(x, \p_t)$, $x \in \UU$.
We will study more general observers in Section \ref{sec_observers}.

\begin{figure}[t!] % float placement: (h)ere, page (t)op, page (b)ottom, other (p)age
  \centering
  \includegraphics[scale=1]{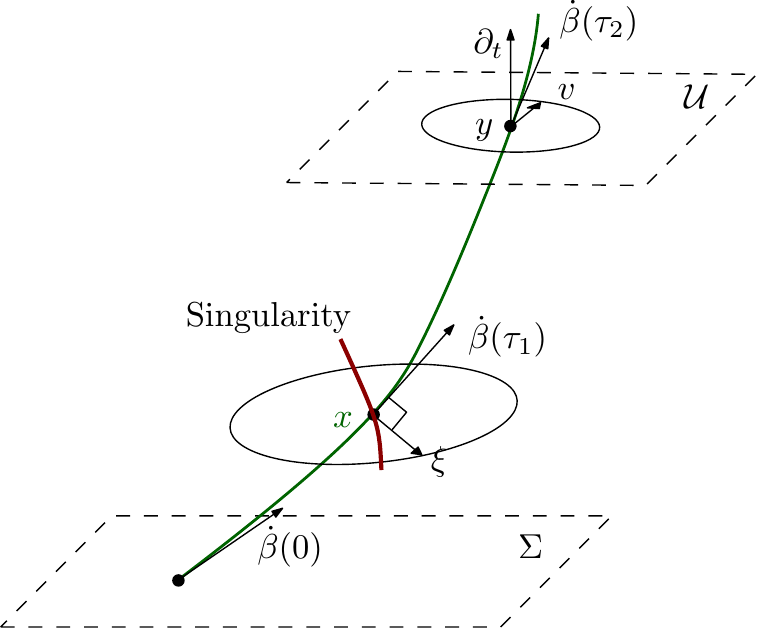}
\caption{The observer $(y, \p_t)$ measures the energy of the photon $\beta$.
The photon has the Newtonian velocity $v$ on the celestial sphere of $(y, \p_t)$.
We show that the measurement carries information on the singularity $(x,\xi) \in \WF(\p_\epsilon g_\epsilon|_{\epsilon = 0})$ that satisfies $\beta(\tau_1) = x$ and $\xi(\dot \beta(\tau_1)) = 0$, see Theorem \ref{th_main}. 
}
\end{figure}

Let $S^2$ be the Euclidean unit sphere in $\R^3$.  We define
the diffeomorphism
\begin{align}
\label{def_theta}
\theta : S^2 \to \mathbb S_{g_\epsilon,x,\p_t} - \p_t,
\quad  
\theta(v) = -(1,a(t_1)^{-1} v) \in \R^{1+3}.
\end{align}
If $g_\epsilon$ is sufficiently close to the background metric $g$
%in the space $C^2_b(M_0;\Sym)$, then
then Lemma \ref{lem_energy} below says that 
$\mathcal S_{g,Z} \UU$ can be parametrized by $\UU_1 \times S^2$,
and the redshifts $\mathcal R_{g_\epsilon, \p_t}$ of CMB photons as measured by $(x,\p_t)$, $x \in \UU$, can be parametrized 
%by the directions in $S^2$ 
as follows
\begin{align}
\label{def_R}
R_\epsilon(y,v) = (\dot \gamma_\epsilon(\tau_\epsilon(x,\theta); x,\theta), \p_t)_{g_\epsilon} - 1,
\quad (y,v) \in \UU_1 \times S^2,
\end{align}
where $x = (t_1, y)$, $\theta = \theta(v)$, 
$\gamma_\epsilon(\cdot; x, \theta)$
is the geodesic of $(M,g_\epsilon)$ with the initial data $(x,\theta)$, 
and 
\begin{align}
\label{def_tau_eps}
\tau_\epsilon(x,\theta) = \min\{\tau > 0;\ 
\gamma_\epsilon(\tau; x, \theta) \in \Sigma \}.
\end{align}
%We say that 
%$$
%\p_\epsilon R_\epsilon(y,v)|_{\epsilon = 0}, \quad (y, v) \in \UU_1 \times S^2,
%$$
%where $R$ is defined by (\ref{def_R}),
%is the linearized CMB measurement on $\UU$.
%Indeed, if a constant $c$ is defined by the formula (\ref{def_reparam}) below and $\theta = \theta(x,v)$, then the reparametrization 
%\begin{align}
%\label{def_reparam_beta}
%\beta_\epsilon(\rho) = \gamma_\epsilon(- c \rho + \tau_\epsilon(x,\theta); x, \theta),
%\end{align}
%of $\gamma_\epsilon(\cdot; x, \theta)$ is a CMB photon satisfying $\beta_\epsilon(\tau_\epsilon(x,\theta)/c) = x$,
%see Lemma \ref{lem_energy} below.

\begin{lemma}
\label{lem_energy}
Let $g$ be a Lorentzian metric tensor of the form (\ref{background_metric}).
Let $t_1 > t_0$, let $\UU_1 \subset \R^3$ be bounded,
and define $\UU = \{t_1\} \times \UU_1$. 
Let $g_\epsilon \in C^2(M; Sym^2)$, $\epsilon \in [0,1]$,
be a one parameter family of Lorentzian metrics
satisfying $g_0 = g$.
Suppose that the map 
$\epsilon \mapsto g_\epsilon$ is differentiable from $[0,1]$
to $C_b^2(M_0; Sym^2)$
and that $g_\epsilon = g$ in $\UU \cup \Sigma$.
Let $(x,v) \in \UU \times S^2$,
and define $\theta = \theta(v)$ and $\tau_\epsilon = \tau_\epsilon(x,\theta)$ for small $\epsilon$. Let $c > 0$.
Then 
the geodesic 
\begin{align*}
\beta_\epsilon(\rho) = \gamma_\epsilon(- c\rho + \tau_\epsilon; x, \theta),
\end{align*}
satisfies $\beta_\epsilon(0) \in \Sigma$ and 
$\beta_\epsilon(\tau_\epsilon/c) = x$.
In particular, if we choose the parametrization
$c = E_0 / (\dot \gamma_\epsilon(\tau_\epsilon; x, \theta), \p_t)_{g_\epsilon}$,
then
$(\dot \beta_\epsilon(0), \p_t)_{g_\epsilon} = -E_0$ 
and the redshift of $\beta_\epsilon$ as measured by $(x,\p_t)$ satisfies
$$
\mathcal R_{g_\epsilon, \p_t}(x, \theta + \p_t)
= (\dot \gamma_\epsilon(\tau_\epsilon; x, \theta), \p_t)_{g_\epsilon} -1.
$$
\end{lemma}
\begin{proof}
We will show that $\tau_\epsilon$
is well-defined for small $\epsilon > 0$
in Lemma \ref{lem_tau_eps} below.
We have $\beta_\epsilon(0) \in \Sigma$, $\beta_\epsilon(\tau_\epsilon/c) = x$,
$(\dot \beta_\epsilon(0), \p_t)_{g_\epsilon} = -E_0$,
%$$
%(\dot \beta_\epsilon(0), \p_t)_{g_\epsilon} = -c 
%(\dot \gamma_\epsilon(\tau_\epsilon; x, \theta), \p_t)_{g_\epsilon} = -\hslash \lambda_0,
%$$
and the redshift is
$$
\frac{(\dot \beta_\epsilon(0), \p_t)_{g_\epsilon}}{(\dot \beta_\epsilon(\tau_\epsilon/c), \p_t)_{g_\epsilon}}
- 1
= 
\frac{-E_0}{-c (\theta, \p_t)_{g_\epsilon}} - 1 
= 
(\dot \gamma_\epsilon(\tau_\epsilon; x, \theta), \p_t)_{g_\epsilon} -1,
$$
since $(\theta, \p_t)_{g_\epsilon} = 1$. 
\end{proof}

%\begin{definition}%[Linearized CMB measurement]
%\label{def_lin_data}
%Let $g$ be a Lorentzian metric tensor of the form (\ref{background_metric}),
%and let $\UU \subset M_0$ be of the form (\ref{def_U})
%where $t_1 > t_0$ and $\UU_1$ is bounded.
%Let $g_\epsilon \in C^2(M; Sym^2)$, $\epsilon \in [0,1]$,
%be a one parameter family of Lorentzian metrics. Suppose that 
%$g_\epsilon = g$ in
% $M \setminus M_0$ and in 
%$\UU$, and that the map 
%$\epsilon \mapsto g_\epsilon$ is differentiable from $[0,1]$
%to $C_b^2(M_0; Sym^2)$. We say that 
%$$
%\p_\epsilon R_\epsilon(x,v)|_{\epsilon = 0}, \quad (x, v) \in \UU \times %S^2,
%$$
%where $R$ is defined by (\ref{def_R}),
%is the linearized CMB measurement on $\UU$.
%\end{definition}

We are ready to formulate the linearized version of Problem \ref{IP} that we will consider.

\begin{problem}[Linearized microlocal CMB inverse problem]
\label{IP'}
Given 
$\p_\epsilon R_\epsilon(x,v)|_{\epsilon = 0}$, $(x, v) \in \UU \times S^2$,
determine the wave front set $\WF(\p_\epsilon g_\epsilon|_{\epsilon = 0})$
on $T^* W$ where $W \subset M$ is as in Problem \ref{IP}.
\end{problem}

We characterize the spacelike vectors in $\WF(\p_\epsilon g_\epsilon|_{\epsilon = 0})$
that can be recovered, and show that 
all timelike vectors are smoothed out.
Physically, the timelike vectors in $\WF(\p_\epsilon g_\epsilon|_{\epsilon = 0})$ correspond to singularities moving faster than light,
whence we do not expect such singularities to be present to begin with. 

Lightlike singularities could correspond to gravitational waves. Their
recovery from
the anisotropies in the Cosmic Microwave Background radiation is an interesting question. However, we leave this as a topic for future work. 

%\textbf{***************}
%
%The analysis below shows that the linearized data remains unchanged if we add a potential field $d^s w$ and a scalar multiple $cg$ of the metric $g$. This is the linearization of the gauge invariance of the non-linear data. Our main result below says that we can recover all spacelike singularities up to that linear subspace. 
%
%\textbf{***************}
%
%the above inverse problem corresponds to a linearization of the problem to determine singularities of $g_\epsilon$ from
%the anisotropies in the energy of the Cosmic Microwave Background radiation.
%In particular, it would be physically interesting if 
%we can detect singularities
%(or topological defects) of the stress-energy tensor (in particular mass) 
%moving slower than light.\footnote{
%Add here a theorem eventually. 
%The right claim seems to be: space-like singularities that change in time 
%can be recovered.
%}

\section{Statement of the results}
\label{sec_results}

Let $f$ be a 
symmetric 2-tensor on a Lorentzian manifold $(M,g)$,
and define the light ray transform of $f$ by
\begin{align}
\label{def_X}
X_{g} f(x,\theta) = \int_\R (f_{lm} \circ \gamma) \dot \gamma^l \dot \gamma^m d\tau, 
\end{align}
where $x \in M$, $\theta \in T_x M$ is lightlike, and $\gamma(\tau) = \gamma(\tau; x, \theta)$
is the geodesic of $(M,g)$ with the initial data $(x,\theta)$.
When 
$g$ is of the form (\ref{background_metric}), we define for $t_1 > t_0$,
$$
L_{t_1}f(y,v) = X_g f(x, \theta), \quad y \in \R^3,\ v \in S^2,
$$ 
where $x = (t_1, y)$ and $\theta = \theta(v)$ is defined by (\ref{def_theta}).
Note that for an arbitrary Lorentzian manifold $(M,g)$, the map $X_g$ may fail to be well defined even when $f \in C_0^\infty(M;\Sym)$. This is the case, for example, if $(M,g)$ has closed null geodesics. 
However, we show the following lemma.

\begin{lemma}
\label{L_is_FIO}
The light ray transform $L_{t_1}$
is a Fourier integral operator of order $-3/4$ whose associated canonical transformation is the twisted conormal bundle of the following point-line relation
$$
\{(p, y, v) \in M \times \R^3 \times S^2;\ p = \gamma(\tau; (t_1, y), \theta(v)) \text{ for some $\tau \in \R$}\}.
$$ 
In particular, $L_{t_1}$ is continuous from $\E'(M;\Sym)$
to $\D'(\R^3 \times S^2)$.

\end{lemma}

The proof of the lemma will be presented in Section \ref{sec_muloc_L}.
The following theorem reduces the linearized CMB inverse problem 
to inversion of a limited angle restriction of the light ray transform.

\begin{theorem}
\label{th_tomography}
Let $t_1 > t_0$, let $\UU_1 \subset \R^3$, and let 
$g_\epsilon \in C^2(M; Sym^2)$, $\epsilon \in [0,1]$,
satisfy the assumptions of Lemma \ref{lem_energy}.
%Suppose that $g_\epsilon = g$ near $\Sigma$.
Define 
$$
f(t,x') = 
\begin{cases}
\p_\epsilon g_\epsilon(t,x')|_{\epsilon = 0}, & t \in (t_0, t_1),
\\
0, & \text{otherwise},
\end{cases}
$$
Then on $\UU_1 \times S^2$,
$$
\p_\epsilon R_\epsilon|_{\epsilon = 0} 
= 
\frac{1}{2 a(t_0)} L_{t_1} a^2 \mathcal L_{a \p_t} a^{-2} f,
$$ 
where  
$\mathcal L_{a \p_t}$ is the Lie derivative along the scaled world velocity $a \p_t$,
and the powers $a^2$ and $a^{-2}$ of the warping factor are interpreted as multiplication operators.
\end{theorem}

The light ray transform has a non-trivial null space.

\begin{lemma}
Let $g$ be a Lorentzian metric tensor of the form (\ref{background_metric})
and let $t_1 > 0$. 
Define $M_1 = (0,t_1) \times \R^3$ and 
$$
\mathcal N = \{ cg + d^s \omega;\ c \in \E'(M_1),\ 
\omega \in \E'(M_1; \L1)\},
$$
where  $d^s$ is the symmetric differential defined in coordinates as follows
\begin{align}
\label{def_potential_field}
(d^s \omega)_{ij} = \frac {(\nabla_i \omega)_j + (\nabla_j \omega)_i} 2.
\end{align}
Here $\nabla_{i} = \nabla_{\p_{x^i}}$
is the covariant derivative on $(M,g)$. 
Then $L_{t_1} f = 0$ for all $f \in \mathcal N$.
\end{lemma}
\begin{proof}
By density of $C_0^\infty$ in $\E'$, we may assume that 
$c \in C^\infty_0(M_1)$ and that 
$\omega \in C^\infty_0(M_1; \L1)$,
and consider %the symmetric 2-tensor 
$f = cg + d^s \omega$.
Then $L_{t_1} f = 0$ on $\R^3 \times S^2$, since
$(\dot \gamma, \dot \gamma)_g = 0$
holds for null geodesics $\gamma$ and 
$$
\p_\tau \omega(\dot \gamma(\tau))
%= \p_t ( \omega_k(\gamma(\tau)) \dot \gamma^k(\tau))
%= (\p_j \omega_k) \dot \gamma^j \dot \gamma^k - \omega_k \Gamma^k_{lm} \dot \gamma^l \dot \gamma^m  
= (d^s \omega)_{ij} \dot \gamma^i(\tau) \dot \gamma^j(\tau)
$$ 
holds for all geodesics $\gamma$.
\end{proof}

%\begin{remark}
In fact, this lemma holds for every Lorentzian metric as far as the null-geodesics are non-trapping. 
Note that the group of transformations $g \mapsto g + h$, $h \in \mathcal N$,
is the linearization of the gauge invariance of the nonlinear problem, see the discussion after Problem \ref{IP}.

The lemma is an analog of the corresponding results for compact Riemannian manifolds, where it is known that the geodesic ray transform vanishes on potential fields, i.e., on tensors of order $m\ge1$ which are symmetric differentials of tensor fields of order $m-1$ vanishing on the boundary, see  \cite{Sharafutdinov1994}. What is new here is the scalar multiple $cg$ of the metric. 

%\end{remark}

Our main result is the following.

\begin{theorem}
\label{th_main}
Let $g$ be a Lorentzian metric tensor of the form (\ref{background_metric}).
Let $t_1 > 0$, let $\UU_1 \subset \R^3$ be open,
and define $M_1 = (0,t_1) \times \R^3$.
Let $(x,\xi) \in T^* M_1$ be spacelike, and
suppose that there is a null geodesic $\gamma$ of $(M,g)$ and $\tau_1, \tau_2\in \R$ such that 
\begin{align}
\label{visibility_cond}
\gamma(\tau_1) = x, \quad \xi(\dot \gamma(\tau_1)) = 0,
\quad \gamma(\tau_2) \in \{t_1\} \times \UU_1.
\end{align}
Then there is $\chi \in S^0(T^*(\R^3 \times S^2))$
supported in $T^*(\UU_1 \times S^2)$
such that 
for all $f \in \E'(M;\Sym)$ the following are equivalent
\begin{itemize}
\item[(i)]  $(x,\xi) \in \WF(L_{t_1}^* \chi L_{t_1} f)$,
\item[(ii)] $(x,\xi) \in \WF(f + h)$
for all $h \in \mathcal N$.
\end{itemize}
%Here $L_{t_1}^*$ is the adjoint of $L_{t_1}$ on $L^2(\UU_1 \times S^2)$.
\end{theorem}

In particular, the theorem says that the operator $L_{t_1}^* \chi L_{t_1}$
does not move spacelike singularities. 
We will give an explicit choice of $\chi$ in Section \ref{sec_microlocal_normalop}
and show that $L_{t_1}^* \chi L_{t_1}$ is a pseudodifferential operator of order $-1$. We will also compute its principal symbol after a conformal scaling, see Proposition \ref{prop_principal_symbol} below.

Note that the visibility condition (\ref{visibility_cond}) is analogous with the 
visibility condition for limited angle X-ray tomography, see e.g. \cite[Th. 3.1]{Quinto1993}.
It is also sharp in the sense that if for a spacelike
$(x,\xi) \in T^* M$ there does not exist a null geodesic $\gamma$
satisfying (\ref{visibility_cond}) then
for all $\chi \in C_0^\infty(\UU_1)$
the wave front set $\WF(\chi L_{t_1} f)$ is empty
if $\WF(f)$ is contained in a small conical neighborhood  
of $(x, \xi)$ in $T^* M$.

We will also show that all timelike singularities are lost in the sense that if $\WF(f)$ contains only timelike covectors then $\WF(L_{t_1} f)$
is empty. In other words, $L_{t_1}$ is smoothing on the cone of timelike covectors.

%We will show below that it is necessary in order to recover a space-like singularity $(x,\xi)$, that is, if (\ref{visibility_cond}) does not hold, and $WF(f)$ is contained in a small conical neighborhood of $(x,\xi)$, then 
%$WF(Xf) = \emptyset$. 

\section{Parametrization of the CMB measurements}
\label{sec_observers}

In this section we show that the function $\tau_\epsilon$, see (\ref{def_tau_eps}),
is well-defined. We do this in a slightly more general context than that in Lemma \ref{lem_energy},
that is, we relax the assumption that $g_\epsilon = g$ on $\UU$, and consider observers $(x,Z_\epsilon)$ on $\UU$ where $Z_\epsilon$ is not necessarily $\p_t$.

\def\WW{\mathcal W}
Let $g_\epsilon \in C^2(M; Sym^2)$, $\epsilon \in [0,1]$,
be a one parameter family of Lorentzian metrics
satisfying $g_0 = g$, where $g$ is a Lorentzian metric tensor of the form (\ref{background_metric}).
Let $\WW_\epsilon \subset M$ be a spacelike submanifold 
of codimension one, and let 
$Z_\epsilon$ be a section of $TM$ defined on $\WW_\epsilon$
such that $(x,Z_\epsilon)$ 
is an observer with respect to $g_\epsilon$.
Let $\UU_1 \subset \R^3$ be open and suppose that 
$\Phi_\epsilon: \UU_1 \to \WW_\epsilon$ is a diffeomorphism.
In this section we consider measurements by the observers
\begin{align}
\label{more_general_observers}
(\Phi_\epsilon(y), Z_\epsilon), \quad y \in \UU_1.
\end{align}
We assume that $\WW_0 = \{t_1\} \times \UU_1$
where $t_1 > t_0$, $\Phi_0(y) = (t_1, y)$,
and that $Z_0 = \p_t$.

For small $\epsilon > 0$, the CMB measurements by observers $(x, Z_\epsilon)$
can be parametrized by the vectors in the celestical spheres
$\mathbb S_{g_\epsilon, x, Z_\epsilon}$, $x \in \WW_\epsilon$.
Indeed, we have the following lemma: 

\begin{lemma}
\label{lem_tau_eps}
Let $M \subset \R^{n}$ be open, let $g$ be a smooth Lorentzian metric tensor on $M$, let $Z$ be a smooth vector field on $M$, and suppose that $(x,Z)$ is an observer for all $x \in M$.  
Let $K_0 \subset M$ be compact
and define 
$$
K = \{(x,V-Z) \in TM;\ x \in K_0,\ V \in \mathbb S_{g,x,Z}\}.
$$
Let $\Sigma \subset M$ 
be a smooth submanifold of codimension one, and suppose that 
the geodesics
$\gamma(\cdot; x, \theta)$, $(x,\theta) \in K$, 
intersect $\Sigma$ non-tangentially. 
Then there are neighborhoods $\mathcal G \subset C_b^2(M; \Sym)$ of $g$ and
$\mathcal K \subset TM$ of $K$ 
such that the geodesics
$\gamma_h(\cdot; x, \theta)$,
$(x,\theta) \in \mathcal K$,
with respect to $h \in \mathcal G$ intersect 
$\Sigma$. Moreover, the function 
$$
\tau_h(x,\theta) = \min \{\tau > 0;\ \gamma_h(\tau; x,\theta) \in \Sigma\} 
$$
is $C^1$ in $\mathcal K \times \mathcal G$,
and the functions $\gamma_h(\tau;x,\theta)$
and $\dot \gamma_h(\tau;x,\theta)$
are $C^1$ in a neighborhood of the set 
$$
\{(\tau, x, \theta, h);\ \tau \in [0, \tau_h(x,\theta)],\ (x,\theta) \in \mathcal K,\ h \in \mathcal G\}.
$$
\end{lemma}
\def\Dom{\text{Dom}}
\begin{proof}
Let $\mathcal G_0 \subset C_b^2(M; \Sym)$ 
be a neighborhood of $g$ such that all $h \in \mathcal G_0$ 
are Lorentzian.
We define the open set $U = TM \times \mathcal G_0$ on the Banach space $E = TM \times C_b^2(M; \Sym)$, and consider the vector field 
$F : U \to E$,
$F(x,\theta, h) = (\theta, f(x, \theta, h), 0)$,
where $f : E \to \R^n$ has components 
$$
f^j(x, \theta, h) = -\Gamma_{k\ell}^j(x,h)\theta^k \theta^\ell, 
\quad (x,\theta) \in TM,\ h \in C_b^2(M; \Sym),
$$
and $\Gamma_{k\ell}^j(x,h)$ are the Christoffel symbols of the
metric tensor $h$ at $x$.
We recall that 
$$
\Gamma_{k\ell}^j(x,h) = 
\frac{1}{2}h^{jm} \left(\frac{\partial h_{mk}}{\partial x^\ell} + \frac{\partial h_{m\ell}}{\partial x^k} - \frac{\partial h_{k\ell}}{\partial x^m} \right),
$$
where $h^{jm}$ is the inverse of $h_{jk}$, and see that 
$\Gamma_{k\ell}^j : M \times \mathcal G_0 \to \R$ 
is continuously differentiable.
Hence $F$ is continuously differentiable, and
the flow of $F$ is continuously differentiable 
on its domain of definition $\Dom(F) \subset \R \times U$,
see e.g. \cite[Th. 6.5.2]{Lang}.
Note that the projection of the flow on $TM$ gives the geodesic flow with respect to the metric tensor $h$.

Let $\kappa = (x,\theta) \in K$.
We have assumed that $(\tau_g(\kappa), \kappa, g) \in \Dom(F)$.
As $\gamma(\cdot;\kappa)$ intersects $\Sigma$ non-tangentially,
the implicit function theorem gives neighborhoods 
$\mathcal K_\kappa \subset TM$
of $\kappa$
and $\mathcal G_\kappa \subset C_b^2(M; \Sym)$ of $g$
such that 
$\tau_h$ is $C^1$ on $\mathcal K_\kappa \times \mathcal G_\kappa$.
The set $K$ is compact since $K_0$ is compact and $S_{g,x,Z}$
is diffeomorphic to the sphere $S^{n-2}$.
We choose a finite set $J \subset K$ such that 
$\mathcal K_\kappa$, $\kappa \in J$, is an open cover of $K$,
and define $\mathcal K = \cup_{\kappa \in J} \mathcal K_\kappa$
and $\mathcal G = \cap_{\kappa \in J} \mathcal G_\kappa$.
\end{proof}

Let us now continue our study of the observers (\ref{more_general_observers}).
We choose a smooth family of diffeomorphisms 
$\Psi_{x,\epsilon} : S^2 \to \mathbb S_{g_\epsilon, x, Z}$, 
$x \in \WW_\epsilon$, such that 
$$
\Psi_{x,0}(v) = (0, -a(t_1)^{-1} v) \in \R^{1+3},
\quad x \in \WW_0.
$$
Then the redshift measurements can be parametrized as
$$
\tilde R_\epsilon(y,v) = (\gamma_\epsilon(\tau_\epsilon(x_\epsilon,\theta_\epsilon); x_\epsilon, \theta_\epsilon), Z_\epsilon)_{g_\epsilon}, \quad (y,v) \in \UU_1 \times S^2,
$$
where $x_\epsilon = \Phi_\epsilon(y)$, $\theta_\epsilon(x,v) = \Psi_{x_\epsilon,\epsilon}(v) - Z_\epsilon$,
$\gamma_\epsilon(\cdot; x, \theta)$ is the geodesic of $(M,g_\epsilon)$
with the initial data $(x, \theta)$,
and $\tau_\epsilon$ is defined by (\ref{def_tau_eps}). 
Note that $\theta_0 = \theta$ where $\theta$ is defined by (\ref{def_theta}).

We omit writing $\epsilon$ as a subscript when it is zero. 
The linearized measurements satisfy 
$$
\p_\epsilon \tilde R_\epsilon(y,v)|_{\epsilon = 0}
= \p_\epsilon R_\epsilon(y,v)|_{\epsilon = 0}
+ \p_\epsilon (\gamma(\tau(x_\epsilon,\theta_\epsilon); x_\epsilon, \theta_\epsilon), Z_\epsilon)_{g}|_{\epsilon = 0},
$$
where the second term is a smooth function on $\UU_1 \times S^2$
if the derivatives at $\epsilon = 0$ of $Z_\epsilon$ and the local parametrizations $\Phi_\epsilon$ and $\Psi_{x,\epsilon}$
 are smooth on $\UU_1 \times S^2$. 
Thus the second term plays no role in the microlocal analysis of $\p_\epsilon \tilde R_\epsilon(x,v)|_{\epsilon = 0}$
and we are reduced to the case covered in Theorem \ref{th_tomography}.

\section{Conformal invariance}

In this section we show that Theorem \ref{th_main}
is invariant under conformal scaling.
Let us recall that a metric tensor $g$ of the form (\ref{background_metric})
is conformal to the Minkowski metric tensor in suitable coordinates. 
Indeed, we define the strictly increasing function 
\begin{align}
\label{conformal_time}
s(t) := \int_0^t a(t)^{-1} dt.
\end{align}
Then 
$$
g(s, x') = a^2(s) (-ds^2 + (dx')^2), \quad (s,x') \in (0,\infty) \times \R^3.
$$
For example, the Einstein-de Sitter model corresponds to $a(t) = t^{2/3}$,
and then
$s(t) = 3 t^{1/3}$ and $a(s) = s^2/9$.

It is well known that null geodesics are invariant up to a reparametrization under a conformal change of the metric. For the reader's convenience we give a proof here. 

\begin{lemma}
\label{lem_confinv_nullgeod}
Let $g$ be a smooth Lorentzian metric tensor, let $a$ be a smooth strictly positive function, and define $\tilde g = a^{-2} g$.
Then the null geodesics $\gamma$ of $g$ 
correspond to the null geodesics $\mu$ of $\tilde g$
under the reparametrization
\begin{align}
\label{reparametrization}
\gamma(\tau) = \mu(\sigma(c \tau)), \quad \tau(\sigma) = \int_0^\sigma 
a(\mu(\rho))^2
%e^{2 \alpha \circ \mu(\rho))} 
d\rho,
\end{align}
where $c = a(\gamma(0))^2$, $\tau(\sigma)$ is a strictly increasing function and $\sigma(\tau)$ is its inverse function.
The reparametrization is chosen so that $\gamma(0) = \mu(0)$ and $\dot \gamma(0) = \dot \mu(0)$.
\end{lemma}
\begin{proof}
We define $\alpha$ by $e^\alpha = a$.
Koszul formula implies 
\begin{align}
\label{covariantd_tildeg}
\nabla_X Y = \tilde \nabla_X Y + (X \alpha) Y + (Y \alpha) X - (X,Y)_{\tilde g} \text{grad}_{\tilde g} \alpha,
\end{align}
where $\text{grad}_{\tilde g}$ is the gradient with respect to the metric $\tilde g$.
We apply this formula to $X = Y = \dot \mu$, where $\mu$ is a null geodesic with respect to $\tilde g$. Then the first and the last term on the right hand side vanish and we have
$$
D_\sigma \dot \mu = 2 (\alpha \circ \mu)' \dot \mu,
$$
where $D_\sigma \dot \mu$ is the covariant derivative of $\dot \mu$ with respect to $g$ along $\mu$, and the prime denotes the derivative of a real valued function on $\R$.
We have $\tau' = e^{2 \alpha \circ \mu}$, $\sigma' = e^{-2 \alpha \circ \mu \circ \sigma}$ and 
\begin{align*}
\sigma'' = -2(\alpha \circ \mu \circ \sigma)' \sigma' = -2((\alpha \circ \mu)' \circ \sigma) (\sigma')^2.
\end{align*}
Let us consider the curve $\gamma$ defined by the formula (\ref{reparametrization}) with $c=1$.
Then the covariant derivative $D_\tau \dot \gamma$ with respect to $g$  along $\gamma$ satisfies
\begin{align*}
D_\tau \dot \gamma &= D_\tau (\sigma' (\dot \mu \circ \sigma))
\\&= \p_\tau (\sigma' (\dot \mu^k \circ \sigma))\p_k
+ (\sigma')^2  (\dot \mu^i \circ \sigma)(\dot \mu^j \circ \sigma) (\Gamma^{k}_{ij} \circ \mu \circ \sigma) \p_k
\\&= \sigma'' (\dot \mu \circ \sigma)
+ (\sigma')^2  (D_\sigma \dot \mu) \circ \sigma = 0.
\end{align*}
Here $\Gamma^{k}_{ij}$ are the Christoffel symbols of $g$.
We see that $\gamma$ is a geodesic with respect to the metric $g$.
The same is true for any $c \ne 0$ since affine reparametrizations of geodesics are geodesics. 
Moreover, 
\begin{align}
\label{conformal_reparametrization}
\dot \gamma(\tau) 
%= \sigma' (\dot \mu \circ \sigma) 
= c a^{-2}(\gamma(\tau)) \dot \mu(\rho), \quad \rho = \sigma(c\tau).
\end{align}
\end{proof}

\begin{corollary}
Let $g$ be a smooth Lorentzian metric tensor, let $a$ be a smooth strictly positive function, and define $\tilde g = a^{-2} g$.
Then the corresponding light ray transforms satisfy
$X_g f = a^2 X_{\tilde g} a^{-2} f$.
%In particular, if $g$ is of the form (\ref{background_metric}), then
%in $(s,x')$ coordinates
%$$
%X_g f(s,x';\theta) = a^{2} X_{\tilde g} (a^{-2} f)(s', x; -1,v),
%$$
%where $\theta = \theta(x,v)$ is defined by  (\ref{def_theta}).
\end{corollary}
\begin{proof}
We set $\rho = \sigma(c \tau)$ and
use (\ref{conformal_reparametrization}) and $d\tau = c^{-1} a^2 d\rho$.
\end{proof}

\begin{lemma}
Let $g$ be a smooth Lorentzian metric tensor, let $a$ be a smooth strictly positive function, and define $\tilde g = a^{-2} g$.
Let $c$ be a smooth function and $\omega$ be a smooth 1-tensor. 
Then 
$$
a^{-2}( c g + d^s \omega ) = \tilde c \tilde g + \tilde d^s \tilde \omega,
$$
where $\tilde c = c + (\tilde \omega, d\log a)_{\tilde g}$
and $\tilde \omega = a^{-2} \omega$.
Here $\tilde d^s$ is the symmetric differential with respect to $\tilde g$.
\end{lemma}
\begin{proof}
We define $\alpha$ by $e^\alpha = a$,
and write (\ref{covariantd_tildeg}) in coordinates.
This gives
the following formula for the Christoffel symbols $\Gamma^k_{ij}$
and $\tilde \Gamma^k_{ij}$ of $g$ and $\tilde g$, respectively,
$$
\Gamma^k_{ij} = \tilde \Gamma^k_{ij}+ \delta^k_i\partial_j\alpha + \delta^k_j\partial_i\alpha-b^k\tilde g_{ij} ,
$$
where $b^k = (\text{grad}_{\tilde g}\alpha)^k$.
The symmetric derivative transforms as
$$
(d^s \omega)_{ij} = 
\frac 1 2 \left( \p_i \omega_j + \p_j \omega_i \right) - \Gamma^k_{ij} \omega_k
= 
(\tilde d^s \omega)_{ij} - \omega_i \partial_j\alpha - \omega_j \partial_i\alpha
+ \omega_k b^k\tilde g_{ij}.
$$
Moreover, 
$$
e^{2 \alpha} (\tilde d^s (e^{-2 \alpha} \omega))_{ij}
= (\tilde d^s \omega)_{ij} - \omega_j \partial_i\alpha - \omega_i \partial_j\alpha.
$$
\end{proof}

Let $g$ be of the form (\ref{background_metric})
and define $\tilde g = a^{-2} g$.
Let $x \in M$ and let $\xi \in T^* M$ be spacelike for $g$.
Then $\xi$ is spacelike for $\tilde g$.
Suppose that $\gamma$ is a null geodesic of $(M,g)$ and that there are $0 < \tau_1 < \tau_2$ such that (\ref{visibility_cond}) holds.
Let $\mu$ be the reparametrization (\ref{reparametrization})
of $\gamma$. Then $\mu$ is a null geodesic of $(M, \tilde g)$
and there are $\rho_1, \rho_2 \in \R$ such that 
\begin{align*}
\mu(\rho_1) = x, \quad \xi(\dot \mu(\rho_1)) = 0,
\quad \gamma(\rho_2) \in \{t_1\} \times \UU_1.
\end{align*}

Suppose now that Theorem \ref{th_main} holds for $\tilde g$, that is,
there is a symbol of order zero $\chi$  
supported in $T^* (\UU_1 \times S^2)$ such that 
for all $\tilde f \in \E'(M;\Sym)$ the following holds:
$$
(x,\xi) \in \WF(\tilde L_{t_1}^* \chi \tilde L_{t_1} \tilde f)
\quad \text{iff}\quad 
(x,\xi) \in \WF(\tilde f + \tilde c \tilde g + \tilde d^s \tilde \omega)
$$
for all $\tilde c \in \D'(M)$ and $\tilde \omega \in \D'(M; \L1)$.
Here $\tilde L_{t_1}(y,v) = X_{\tilde g}(x, \theta)$,
where $x = (t_1, y)$ and $\theta = \theta(v)$ is defined by (\ref{def_theta}). 
Note that  $L_{t_1} =a^2(t_1) \tilde L_{t_1} a^{-2}$, and 
$$ 
L_{t_1}^* \chi L_{t_1} = a^4(t_1) a^{-2} \tilde L_{t_1}^* \chi \tilde L_{t_1} a^{-2}.
$$

Let $f \in \E'(M;\Sym)$, $c \in \D'(M)$
and $\omega \in \D'(M; \L1)$, and let us define $\tilde f = a^{-2} f$,
$\tilde \omega = a^{-2} \omega$, and 
$\tilde c = c + (\tilde \omega, d \log a)_{\tilde g}$.
Then 
$$
a^{-2}(f + cg + d^s \omega) = \tilde f + \tilde c \tilde g + \tilde d^s \tilde \omega,
$$
and  
%$L_{t_1}^* \chi L_{t_1} f = a^4(t_1) a^{-2} \tilde L_{t_1}^* \chi \tilde L_{t_1} \tilde f$
\begin{align*}
(x,\xi) \in \WF(L_{t_1}^* \chi L_{t_1} f),
\quad &\text{iff} \quad
(x,\xi) \in \WF(\tilde L_{t_1}^* \chi \tilde L_{t_1} \tilde f)
\\
(x,\xi) \in \WF(f + cg + d^s \omega)
\quad &\text{iff} \quad
(x,\xi) \in \WF(\tilde f + \tilde c \tilde g + \tilde d^s \tilde \omega).
\end{align*}
Thus Theorem \ref{th_main} holds for $g$ if it holds for $\tilde g$.

\section{Derivation of the tomography problem}
\def\ll{\left}
\def\rr{\right}

In this section we prove Theorem \ref{th_tomography}.
The proof is similar to that in \cite{SWo}.
All the differentiations below are justified by Lemma \ref{lem_tau_eps}.
We omit writing $\epsilon = 0$ as a subscript.

Let $t_1 > t_0$, let $\UU_1 \subset \R^3$, and 
define $\UU = \{t_1\} \times \UU_1$.
Let $(x,v) \in \UU \times S^2$ and write $\theta = \theta(x,v)$.
Note that $\p_t = a^{-1} \p_s$ where $s$ is defined by (\ref{conformal_time}). 
We have
$$
\p_\epsilon (\dot \gamma_\epsilon(\tau_\epsilon(x,v); x, \theta), \p_t)_{g_\epsilon}
= 
a^{-1}(t_0) \p_\epsilon (\dot \gamma_\epsilon(\tau(x,v); x, \theta), \p_s)_{g_\epsilon},
$$
since $\ddot \gamma = 0$ and therefore
\begin{align*}
\p_\epsilon (\dot \gamma(\tau_\epsilon(x,v); x, \theta), \p_s)_{g}
&=
\p_\tau (\dot \gamma(\tau; x, \theta), \p_s)_{g}|_{\tau = \tau_\epsilon(x,v)}
\p_\epsilon \tau_\epsilon(x,v)
\\&=
(\ddot \gamma(\tau_\epsilon(x,v); x, \theta), \p_s)_{g}\p_\epsilon \tau_\epsilon(x,v)
= 0.
\end{align*}
Moreover, 
$$
\p_\epsilon (\dot \gamma_\epsilon(0; x, \theta), \p_s)_{g_\epsilon}
= \p_\epsilon (\theta, \p_s)_g = 0,
$$
since $g_\epsilon = g$ in $\UU$.
Hence
\begin{align}
\label{data_as_diff}
a(t_0) \p_\epsilon R(x,v) = \p_\epsilon \left( (\dot \gamma_\epsilon(\tau(x,v); x, \theta), \p_s)_{g_\epsilon} - (\dot \gamma_\epsilon(0; x, \theta), \p_s)_{g_\epsilon} \right).
\end{align}

We use the reparametrization (\ref{reparametrization})
and recall the formula (\ref{conformal_reparametrization}). Then
\begin{align}
\label{gamma_mu}
(\dot \gamma_\epsilon(\tau; x, \theta), \p_s)_{g_\epsilon}
=
c (\dot \mu_\epsilon(\rho; x, \theta), \p_s)_{\tilde g_\epsilon},
\end{align}
where $c = a(x)^2$, $\rho = \sigma(c\tau)$, $\tilde g_\epsilon = a^{-2} g_\epsilon$
and $\mu_\epsilon(\rho; x, \theta)$ is the geodesic of $(M,\tilde g_\epsilon)$
with initial data $(x, \theta)$. 

The computations below will use the fact that 
$\tilde g = -ds^2 + (dx')^2$
is a constant tensor in the $(s,x')$ coordinates.
When computing in coordinates 
we write $(s,x') = (x^0, x^1, x^2, x^3)$,
and use also the notation $\mu_\epsilon(\rho) = \mu_\epsilon(\rho; x, \theta)$.

The geodesic equations for $\mu_\epsilon$ can be written in the form
\begin{align}
\label{geod_eq}
\p_\rho (\tilde g_{\epsilon,jk} \dot \mu_\epsilon^k) 
=  \frac 1 2 \ll(\p_{x^j} \tilde g_{\epsilon,lm} \rr) \dot \mu_\epsilon^l \dot \mu_\epsilon^m.
\end{align}
%where $\tilde g_\epsilon$ and its derivative are evaluated at $\mu_\epsilon$.
Indeed, the equation (\ref{geod_eq}) can be transformed to the equation,
$$
\p_\rho^2 \mu_\epsilon^j + \tilde \Gamma^j_{\epsilon,lm} \dot \mu_\epsilon^l \dot \mu_\epsilon^m = 0,
$$
written in terms of the Christoffel symbols $\tilde \Gamma^j_{\epsilon,lm}$ of the metric 
$\tilde g_\epsilon$,
by using the chain rule on the left hand side and symmetrizing the obtained sum, 
see e.g. \cite[p. 51]{Taylor1996}.

We denote $h = \p_\epsilon \tilde g_\epsilon|_{\epsilon = 0}$
and have
%$$
%\p_\epsilon (\tilde g_{\epsilon} \circ \mu_\epsilon)|_{\epsilon = 0} 
%= \p_\epsilon (\tilde g_{\epsilon} \circ \mu)|_{\epsilon = 0} + \p_\epsilon (\tilde g \circ \mu_\epsilon)|_{\epsilon = 0}
%= \p_\epsilon \tilde g_{\epsilon}(\mu)|_{\epsilon = 0},
%= h \circ \mu,
%$$
%since $\tilde g$ is a constant tensor. Moreover,
\begin{align*}
\p_\epsilon \ll( \ll(\p_{x^j} \tilde g_{\epsilon,lm} \rr) \dot \mu_\epsilon^l \dot \mu_\epsilon^m \rr)|_{\epsilon = 0}
&= \ll(\p_{x^j} \p_\epsilon \tilde g_{\epsilon,lm} |_{\epsilon = 0} \rr) \dot \mu^l \dot \mu^m
+ \ll(\p_{x^j} \tilde g_{lm} \rr) \p_\epsilon \ll(\dot \mu_\epsilon^l \dot \mu_\epsilon^m \rr)|_{\epsilon = 0}
\\&= \ll( \p_{x^j} h_{lm} \rr) \dot \mu^l \dot \mu^m,
\end{align*}
where we have used the fact that $\tilde g$ is a constant tensor.
Hence 
\begin{align*}
%\label{to_be_integrated}
\p_\rho \p_\epsilon (\dot \mu_\epsilon, \p_s)_{\tilde g_\epsilon}|_{\epsilon = 0} 
= \p_\epsilon \p_\rho \ll( \tilde g_{\epsilon,0k} \dot \mu_\epsilon^k \rr)|_{\epsilon = 0}
= \frac 1 2 \ll( \p_{x^0} h_{lm} \rr) \dot \mu^l \dot \mu^m.
\end{align*}
We integrate this with respect to $\rho$ and obtain
\begin{align}
\label{integration_of_geod}
&\frac 1 2 \int _0^r (\p_{s} h_{lm}) \dot \mu^l \dot \mu^m d\rho 
= 
\p_\epsilon (\dot \mu_\epsilon(r), \p_s)_{\tilde g_\epsilon}|_{\epsilon = 0}
- \p_\epsilon (\dot \mu_\epsilon(0), \p_s)_{\tilde g_\epsilon}|_{\epsilon = 0},	
\end{align}
where $r > 0$.
Now (\ref{integration_of_geod}), (\ref{data_as_diff}) and (\ref{gamma_mu}) imply
$$
c^{-1} a(t_0) \p_\epsilon R(x,v)|_{\epsilon = 0}
= \frac 1 2 \int _0^{\sigma_0} (\p_{s} h_{lm}) \dot \mu^l \dot \mu^m d\rho,
$$
where $\sigma_0 = \sigma(c\tau(x,v))$.

Let $t_1 > 0$, let $f$ be the cutoff of $\p_\epsilon g_\epsilon|_{\epsilon = 0}$ as in Theorem \ref{th_tomography},
and suppose that $x \in \UU_1$. 
Then $h = a^{-2} f$ on $\mu([0, \sigma_0])$ and $f = 0$
on $\mu(\R \setminus [0, \sigma_0])$.
Note also that 
$
(\mathcal L_{a \p_t} h)_{lm} = (\mathcal L_{\p_s} h)_{lm} = \p_s h_{lm}.
$
Hence 
$$
\p_\epsilon R(x,v)|_{\epsilon = 0}
= \frac{c}{2 a(t_0)} X_{\tilde g} \mathcal L_{a \p_t} a^{-2} f
= \frac{1}{2 a(t_0)} X_g a^2 \mathcal L_{a \p_t} a^{-2} f,
$$
and we have proven Theorem \ref{th_tomography}.

\section{Backprojection of the light ray transform}
\label{sec_backprojection}

\def\S{{S^2}}
Let us now turn to the proof of Theorem \ref{th_main}.
By the conformal invariance, we may assume that the background metric is the Minkowski metric, that is, $a = 1$ identically. 
If we have data on all the light rays, i.e.,  if $\UU_1 = \R^3$,
then the light ray transform $L_{t_1}$ is invariant with respect to translations. 
Although the case $\UU_1 = \R^3$ is physically unrealistic, 
we discuss it briefly in this section since it allows for very explicit computations.

After a translation in the $t$ coordinate, 
we may suppose that $t_1 = 0$. 
We write $L = L_{0}$, and have
\begin{align}
\label{def_L_translated}
Lf(y,v) = \int_\R f_{lm}(s, y + sv) \theta^l \theta^m ds,
\quad y \in \R^3,\ v \in S^2,
\end{align}
where $\theta = (1, v) \in \R^{1+3}$. Note that 
we have made change of variable $s = -\tau$
in comparison with (\ref{def_X}),
and also use $\theta = -\theta(v)$ in comparison with (\ref{def_theta}).

We use the Euclidean surface measures in all the integrations below.
Let $f, h \in C_0^\infty(\R^4; \Sym)$. Then 
\begin{align*}
&(Lf, Lh)_{L^2(\R^3 \times S^2)}
\\&\quad= \int_{S^2} \int_\R \int_{\R^3} \int_\R
f_{lm}(r,y+rv) h_{jk} (s,y+sv) \theta^j \theta^k \theta^l \theta^m
ds dy dr dv
\\&\quad= \int_{S^2} \int_\R \int_{\R^4}
f_{lm}(r,x'+(r-x^0)v) h_{jk} (x) \theta^j \theta^k \theta^l \theta^m
dx dr dv,
\end{align*}
where we have used the fact that
\begin{align}
\label{change_of_var_sy}
(s, y) \mapsto (s, y + s v) = x = (x^0, x')
\end{align}
is a linear isometry on $\R^{1+3}$.
After making the change of variables $\rho = x^0 - r$, we see that
%$$
%(L^* L f)^{jk}(x) 
%= \int_{S^2} \int_\R 
%f_{lm}(x^0-\rho,x'-\rho v) \theta^j \theta^k \theta^l \theta^m
%d\rho dv,
%$$
and the normal operator $L^* L$ is the
convolution $K^{jklm} * f_{lm}$, 
with the kernel
\begin{align*}
(K^{jklm}, \phi)_{\D' \times \D(\R^4)} 
%&= 
%\int_{S^3} \int_0^\infty \kappa^* \delta(\theta) \theta^l \theta^m \theta^j \theta^k \phi(\rho \theta) d\rho d\theta
&=
\int_{\S} \int_\R \theta^j \theta^k \theta^l \theta^m \phi(\rho \theta) d\rho dv.	
\end{align*}

\def\SS{\mathcal S}

\def\aa{\mathfrak a}

%\HOX{Check that the 1st ref should be Sharafutdinov's book. This was missing in the file. -Lauri}
The following lemma is in the spirit of the representation of the ray transform in 
\cite{Sharafutdinov1994, SU-Duke} if we allow for a singular weight on the light cone. 

\begin{lemma}
\label{lem_Fourier}
The Fourier transform of $K^{jklm}$ is a locally integrable function and satisfies
$$
\widehat K^{jklm}(\xi) =
\begin{cases}
 2 \pi (|\xi'|^2 - |\xi_0|^2)^{-1/2}\int_{S^1_\xi} \theta^j \theta^k \theta^l \theta^m dv, & \text{$\xi$ is spacelike},
\\ 0, & \text{otherwise}.
\end{cases}
$$
Here $\xi = (\xi_0, \xi') \in \R^{1+3}$, $\theta = (1,v)$, $|\cdot|$ is the Euclidean norm,
and
\begin{equation}\label{S1}
S^1_\xi = \{v \in S^2;\ \xi_0 + \xi' v = 0 \}
\end{equation}
 is a circle of radius 
$|\xi'|^{-1} \sqrt{|\xi'|^2 - |\xi_0|^2}$.
\end{lemma}
\begin{proof}
To simplify the notation, we fix the indices $j,k,l,m$ and denote $\aa(v) = \theta^j \theta^k \theta^l \theta^m $.
Let us also define $\kappa(v,\xi) = \xi_0 + \xi' v$. 
Then
\begin{align*}
(\widehat K^{jklm}, \phi)_{\SS'\times \SS(\R^4)}
&= \int_{\S} \int_\R \int_{\R^4} \aa e^{-i\rho \xi \cdot \theta} \phi(\xi) d\xi d\rho dv
\\&= 2 \pi \int_{\S} (\kappa^* \delta, \phi)_{\SS'\times\SS(\R^4)}\, \aa dv.
%\\&= c ((\kappa^* \delta, a)_{\D' \times \D(\S)}, \phi)_{\SS'\times\SS(\R^4)},
\end{align*}
Here the pullback is 
by $\kappa_v(\xi) = \kappa(v,\xi)$ with fixed $v$.
Note that the gradient $\nabla \kappa_v = \theta$ does not vanish,
and therefore the pullback is well-defined.

Let us assume for a moment that $\supp(\phi)$
does not intersect the set of lightlike vectors 
$$
\{\xi = (\xi_0, \xi') \in \R \times \R^3;\ |\xi_0| = |\xi'|\}.
$$
Then
$$
\int_{\S} (\kappa^* \delta, \phi)_{\SS'\times\SS(\R^4)}\, \aa dv
= ((\kappa^* \delta, \aa)_{\D' \times \D(\S)}, \phi)_{\SS'\times\SS(\R^4)},
$$
Here the second pullback is by $\kappa_\xi(v) = \kappa(v,\xi)$ with fixed $\xi$.
Let us now show that this is well-defined.
We %write $\xi = (\xi_0, \xi') \in \R \times \R^3$ and
have $\nabla_{\R^3} \kappa_\xi = \xi'$. 
Thus on $S_\xi^1 = \kappa_\xi^{-1}(0)$
$$
|\nabla_{\S} \kappa_\xi(v)|^2 
= |\xi' - (\xi' v) v|^2 = 
|\xi'|^2 - |\xi'v|^2 = 
|\xi'|^2 - |\xi_0|^2.
$$
We see that the pullback by $\kappa_\xi$ is well-defined if and only if $\xi$ is not lightlike.

If $\xi$ is timelike, then 
$|\xi'| < |\xi_0| = |\xi'v| \le |\xi'|$ on $S_\xi^1$ 
which is a contradiction. Thus $S_\xi^1$ does not intersect the timelike vectors
and $\widehat  K^{jklm} = 0$ on timelike vectors.

Suppose now that $\xi$ is spacelike. Then
\begin{align}
\label{kappaxi_pullback}
(\kappa_\xi^* \delta, \aa)_{\D' \times \D(\S)} = (|\xi'|^2 - |\xi_0|^2)^{-1/2} 
\int_{S_\xi^1} \aa(v) dv.
\end{align}
The set $S_\xi^1$ is the intersection of $\S$ with the affine plane 
\begin{align*}
%\label{def_affplane_P}
P_\xi = \{v \in \R^3;\ \xi'v = - \xi_0\}.
\end{align*}
As $\xi$ is spacelike, $0 \le |\xi_0| < |\xi'|$,
and we see that $w = -\xi_0\xi'/|\xi'|^2 \in P_\xi$.
Thus $S_\xi^1$ is a circle of radius $r$ satisfying $r^2 + |w|^2 = 1$,
that is,
$$
r = |\xi'|^{-1} \sqrt{|\xi'|^2 - |\xi_0|^2}.
$$

To finish the proof, it is enough to show that $\widehat K^{jklm}$
coincides with a locally integrable function.
Let us first assume that $\supp(\phi)$ does not intersect the set
\begin{align}
\label{set_xiprime_zero}
\{\xi = (\xi_0, \xi') \in \R \times \R^3;\ |\xi'| = 0\}.
\end{align}
Then for each $\xi \in \supp(\phi)$
we can choose spherical coordinates so that
$$
v = v(\alpha, \beta) = (\sin \alpha \cos \beta, \sin \alpha \sin \beta, \cos \alpha)
$$
and $v(0, \beta) = -\xi'/|\xi'|$. In these coordinates 
\begin{align*}
\int_{\S} (\kappa^* \delta, \phi)_{\SS'\times\SS(\R^4)}\, \aa dv
&=
\int_{S^2} \int_{\R^3} \phi(-\xi'v, \xi') \aa(v) dv d\xi
\\&=
\int_{\R^3} \int_0^{2\pi} \int_0^\pi \phi(|\xi'| \cos \alpha, \xi') \aa(\alpha, \beta)
\sin \alpha d\alpha d\beta d\xi.
\end{align*}
The change of coordinates $\xi_0 = |\xi'| \cos \alpha$
gives 
\begin{align*}
\int_{\S} (\kappa^* \delta, \phi)_{\SS'\times\SS(\R^4)}\, \aa dv
&=
\int_{\R^3} |\xi'|^{-1} \int_0^{2\pi} \int_{-|\xi'|}^{|\xi'|} \phi(\xi_0, \xi') \aa(\xi_0, \beta) d\xi_0 d\beta d\xi.
\end{align*}
Hence away from the set (\ref{set_xiprime_zero}), the Fourier transform $\widehat K^{jklm}$
is the function
\begin{align}
\label{symbol_via_sphericalcoord}
|\xi'|^{-1} 1_{C}(\xi) \int_0^{2\pi} \aa(\xi_0, \beta) d\beta,
\end{align}
where $1_C$ is the indicator function of the set of spacelike vectors, that is, $1_C(\xi) = 1$ if $\xi$ is spacelike and it is zero otherwise.
Note that $\aa$ is bounded and that the function (\ref{symbol_via_sphericalcoord})
is locally integrable.

We have seen that $\widehat K^{jklm}$ vanishes on timelike vectors
and that it is a function away from the set (\ref{set_xiprime_zero}).
The origin is the only vector in the set (\ref{set_xiprime_zero}) 
that is not timelike. 
As $K^{jklm}$ is homogeneous of degree $-3$, its Fourier transform is homogeneous of degree $-1$.
The Fourier transform $\widehat K^{jklm}$ is a function since 
the only distributions supported in the origin are the linear combinations of derivatives of the delta distribution, and these derivatives are homogeneous of degree $-4$ or less.
\end{proof}

Note that the integral 
$
\int_{S^1_\xi} \theta^j \theta^k \theta^l \theta^m dv
$
is homogeneous of degree zero with respect to $\xi$.
It can be evaluated explicitly for all the $4^4$ combinations of the indices.
As an example, let us give $\widehat K^{jjjj}(\xi)$, $j=0,1,2,3$,
for a spacelike $\xi \in \R^4$. After a rotation in $\xi'$, we may assume that 
$\xi = (\xi_0, 0, 0, -|\xi'|)$.
Then at $\xi$ it holds that 
$$
\widehat K^{0000}(\xi) = \frac{4 \pi^2}{|\xi'|},
\quad \widehat K^{1111}(\xi)
= \frac{3\pi^2  (|\xi'|^2 - |\xi_0|^2)^2}{2 |\xi'|^5},
\quad \widehat K^{3333}(\xi) = \frac{4 \pi^2 \xi_0^4}{|\xi'|^5},
$$
and $\widehat K^{2222}(\xi) = \widehat K^{1111}(\xi)$.

An analogous light ray transform can be defined for scalar functions as follows,
$$
L_{sc} f(y, v) = \int_\R f(s,y+sv) ds, \quad y \in \R^3,\ v \in S^2,
$$
and a variation of the above argument shows that $L_{sc}^* L_{sc}$
is given by the Fourier multiplier $\widehat K^{0000}$.
The scalar case will be considered in depth in \cite{SU_book_in_prog}.

\section{The null space of the symbol of the normal operator}

In this section we will characterize the null space of the 
tensor $\widehat K^{jklm}(\xi)$
for spacelike $\xi \in \R^{4}$.
Note that if $\omega \in C_0^\infty(\R^4; \L1)$, then the Fourier transform of its symmetric differential is 
$$
(\widehat{d^s \omega })_{ij}(\xi) =
\frac {\xi_i \widehat \omega _j + \xi_j \widehat \omega_i} 2.
$$
To cover also the limited angle case, to be discussed in the next two sections,
we will prove the following slightly more general result. We refer to (\ref{S1}) for the definition of $S^1_\xi$. 

\def\ker{\mathop{\rm Ker}}
\begin{lemma}
\label{lem_nullspace}
Let $\xi \in \R^{4}$ be spacelike, let $\chi \in C^\infty(S_\xi^1)$ 
and suppose that $\chi \ge 0$ and that $\chi$ does not vanish identically.
Then the null space of the  linear map $N:f_{lm} \mapsto N^{jklm} f_{lm}$ on 
$\Sym$, where 
\begin{align}
\label{def_N}
N^{jklm} = \int_{S^1_\xi} \chi \theta^j \theta^k \theta^l \theta^m dv,
\quad \theta = (1,v),\ j,k=0,1,2,3,
\end{align}
 is
$
\ker(N) = \{c g_{lm} + \xi_l \omega_m + \xi_m \omega_l;\ c \in \R,\ \omega \in \R^4 \}.
$
Here $g$ is the Minkowski metric.
\end{lemma} 
\begin{proof}
We have $g_{lm} \theta^l \theta^m = 0$ since $\theta$ is lightlike,
and $\theta^m \xi_m = 0$ since $v \in S_\xi^1$.
Thus the tensors $cg + \xi \otimes \omega + \omega \otimes \xi$ are in $\ker(N)$. 
Let us now show that there are no other tensors in $\ker(N)$.
Suppose that $f \in \ker(N)$. Then
$$
0 = \int_{S^1_\xi} \chi  \bar f_{jk}\theta^j \theta^k \theta^l \theta^m f_{lm}dv
= \int_{S^1_\xi} \chi |f_{jk} \theta^j \theta^k|^2 dv.
$$
Then 
$f_{jk} \theta^j \theta^k = 0$ for $v$ in a nonempty open subset of 
$S^1_\xi$.

We write $\xi = (\xi_0, \xi') \in \R^{1+3}$.
After a rotation in $\xi'$, we may assume that $\xi' = (0,0,1)$.
As $\xi$ is spacelike, there is a 
Lorentz boost $B$ such that $\xi = B e_3$,
where $e_3 = (0,0,0,1)$.
The boost $B$ is a linear map $B^k_j$ such that if $\tilde \theta^k = B^k_j \theta^j$
then $\tilde \theta^j = \theta^j$, $j=1,2$, and
\begin{align}
\label{boost}
\tilde \theta^0 = \theta^0 \cosh \alpha + \theta^3 \sinh \alpha,
\quad
\tilde \theta^3 = \theta^0 \sinh \alpha + \theta^3 \cosh \alpha,
\end{align}
where $\alpha$ is the hyperbolic angle (or rapidity) of the boost.

When $v \in S_\xi^1$ and $\theta = (1,v)$ we have
$$
0 = \xi \theta = (B e_3) \theta = \tilde \theta^3,
$$
since the boost $B$ is symmetric.
Combining this with (\ref{boost}) and $\theta^0 = 1$
gives $\tilde \theta^0 = 1/\cosh \alpha$.
Indeed,
$$
\theta^3 = - \frac {\sinh \alpha} {\cosh \alpha} \quad  \text{and} \quad
\tilde \theta^0 = \cosh \alpha - \frac {\sinh^2 \alpha} {\cosh \alpha} = \frac {1} {\cosh \alpha}.
$$

Let us denote $A = B^{-1}$ and define 
a $3 \times 3$ symmetric tensor $h$ by 
$$
h_{lm} = \cosh^{-2} \alpha\,\tilde f_{lm},
\quad l,m=0,1,2,
$$
where $\tilde f_{lm} = A^j_l f_{jk} A^k_m$.
Let us also write $W = \cosh \alpha (\tilde \theta^0, \tilde \theta^1, \tilde \theta^2) \in \R^3$
and $w = (W^1, W^2) \in \R^2$.
Then
$$
0 = f_{jk} \theta^j \theta^k = \tilde f_{lm} \tilde \theta^l  \tilde \theta^m = h_{lm} W^l  W^m,
$$
since $\tilde \theta^3 = 0$.
Moreover, $W^0 = 1$ and $w = \cosh \alpha (\theta^1, \theta^2)$ satisfies
$$
|w|^2 = \cosh^2 \alpha (1 - |\theta^3|) = \cosh^2 - \sinh^2 \alpha = 1,
$$
since $\theta = (1,v)$ and $v \in S^2$.
Thus 
$$
h((1,w), (1,w)) = h_{lm} W^l  W^m = 0
$$ for $w$ in a nonempty open subset of $S^1$.
After a rotation we may assume that $w = (\sin \beta, \cos \beta)$ is in this set for small $\beta$.
Differentiating three times with respect to $\beta$ we get 
\begin{align}
\label{d0}
0 &= h((1,w), (1,w)),
\\
\label{d1}
0 &= h((1,w), (0,w')),
\\
\label{d2}
0 &= h((0,w'), (0,w')) - h((1,w), (0,w)),
\\
\label{d3}
0 &= -3h((0,w'), (0,w)).
\end{align}
Now (\ref{d3}) implies that $h_{12} = 0$. By (\ref{d3}) we have
$$
h((1,w), (0,w')) = h((1,0),(0,w')) + h((0,w),(0,w'))
= h((1,0),(0,w')),
$$
which implies together with (\ref{d1}) that $h_{02} = 0$.
Differentiating this once more we see that $h_{01} = 0$.
As all the cross terms vanish, (\ref{d0}) implies that 
$h_{11} = -h_{00}$. Differentiating (\ref{d3}) we see that
$$
0 = h((0,w),(0,w)) - h((0,w'),(0,w')),
$$
which implies that $h_{11} = h_{22}$. Thus $h$ is proportional to the $1+2$ dimensional Minkowski metric.

We have shown that there is $c \in \R$ such that 
$\tilde f_{ml} - c g_{ml}$
vanishes when $m,l = 0,1,2$. Moreover, 
we can choose $\tilde \omega \in \R^4$ such that 
$$
\tilde f - c g - e_3 \otimes \tilde \omega  - \tilde \omega  \otimes e_3 = 0.
$$
We recall that $\xi = B e_3$.
As the $1+3$ dimensional Minkowski metric $g$ is invariant under the Lorentz boost $B$, we have
$$
f = cg + \xi \otimes \omega + \omega \otimes \xi,
$$
where $\omega = B \tilde \omega$.
\end{proof}

Let us now construct a projection onto a complement of $\ker(N)$ as a symbol $P$ of order zero that commutes with $N$. 
The construction is inspired by \cite{Sharafutdinov1994, SU-Duke}.
We will raise and lower indices by using the Minkowski metric $g$. 
We define 
$$
\lambda^{ij}_{kl} = \kappa^i_k \kappa^j_l,
\quad \kappa^i_k = \delta^i_k - \eta_k \eta^i,
\quad \eta = \frac \xi {\sqrt{g(\xi, \xi)}},
$$ 
where $\xi$ is assumed to be spacelike,
and $\delta^i_k$ is the Kronecker delta. Then $\eta^i \eta_i = g(\eta, \eta) = 1$ and 
\begin{align}
\label{kappa_eta}
\kappa^i_k \eta_i = \eta_k - \eta_k \eta^i \eta_i = 0.
\end{align}
Thus $\xi \otimes \omega + \omega \otimes \xi \in \ker(\lambda)$ for any $\omega \in \R^4$.

We have
$\kappa^i_k g_{ij} = g_{kj} - \eta_k \eta_j$,
whence $\lambda g = \mu$, where
$$
\mu_{kl} = g_{kl} - \eta_k \eta_l.
$$ 
Let us define $h^{ij} = g^{ij} - \eta^i \eta^j$, and 
\begin{align}
\label{the_projection}
P = \lambda - \frac{\mu \otimes h} 3.
\end{align}

\begin{lemma}
Let $N$ be a tensor of the form (\ref{def_N}).
Then $P$ is a projection satisfying $\ker(P) = \ker(N)$ and $PN = N = NP$.
\end{lemma}
\begin{proof}
We have
$$
h^{ij} g_{ij} = g^{ij} g_{ij} - g(\eta, \eta) = 4 - 1 = 3,
$$
whence $g \in \ker(P)$. Moreover, 
\begin{align}
\label{h_eta}
h^{ij} \eta_i = \eta^j - \eta^j \eta^i \eta_i =0, 
\end{align}
whence 
$\xi \otimes \omega + \omega \otimes \xi \in \ker(P)$ for any $\omega \in \R^4$.
We have shown that $\ker(N) \subset \ker(P)$.

Let us show that $P^2 = P$. We have 
$\kappa_k^p \kappa_p^i = \kappa_k^i$
by (\ref{kappa_eta}), whence $\lambda^2 = \lambda$.
Similarly (\ref{h_eta}) implies that 
$$
h^{ij} \mu_{ij} = \delta^{ij}_{ij} - g(\eta, \eta) = 3,
$$
and $(\mu \otimes h)^2 = 3 \mu \otimes h$.
Moreover,
$$
\lambda (\mu \otimes h) = (\lambda \mu) \otimes h = (\lambda g) \otimes h = \mu \otimes h.
$$
Analogously with above, we have %$h \lambda = h$ and
$\mu \otimes h \lambda = \mu \otimes h$. Thus
$$
P^2 = \lambda^2 - \lambda \frac{\mu \otimes h} 3
- \frac{\mu \otimes h} 3 \lambda + \frac{(\mu \otimes h)^2} 9
= \lambda - 2\frac{\mu \otimes h} 3
+ \frac{\mu \otimes h} 3 = P.
$$

We have $P_{kl}^{pq} N_{pq}^{ij} = N_{kl}^{ij}$ since 
$$
\kappa_k^p \theta_p = \theta_k - \eta_k \frac {\xi^p \theta_p} {\sqrt{g(\xi, \xi)}} = \theta_k,
\quad
h^{pq} \theta_p \theta_q = 
g(\theta, \theta) - \frac{\xi^p \theta_p \xi^q \theta_q}{g(\xi, \xi)} = 0.
$$
Analogously $NP = N$, therefore $\ker(P) \subset \ker(N)$.
\end{proof}

\begin{lemma}
\label{lem_microlocal_nullsp}
Let $M \subset \R^4$ be open and let $(x,\xi) \in T^* M$ be spacelike.
Then for all $f \in \E'(M; \Sym)$ we have that
$(x,\xi) \in \WF(P f)$
if and only if
$(x,\xi) \in \WF(f + c g + d^s \omega)$
for all $c \in \E'(M)$ and $\omega \in \E'(M; \L1)$.
Furthermore, if $A$ is a pseudodifferential operator 
satisfying $A(c g + d^s \omega) = 0$
for all $c \in \E'(M)$ and $\omega \in \E'(M; \L1)$,
then 
$
(x,\xi) \in \WF(A P f)$
if and only if
$(x,\xi) \in \WF(Af)$.
\end{lemma}
\begin{proof}
Let $\chi \in S^0(T^* M)$
satisfy $\chi(x,\xi) = 1$ and suppose that $\supp(\chi)$
contains only spacelike covectors.
Note that $(x,\xi) \in \WF(Pf)$
if and only if $(x,\xi) \in \WF(P \chi f)$,
since $P (1-\chi)$ is smoothing near $(x,\xi)$.

Let $h$ be the Fourier transform of $(1-P)\chi f$.
Then $h(\eta) \in \ker(P) = \ker(N)$,
and therefore $h(\eta)$ is of the form
$$
\tilde c(\eta)g + \eta \otimes \tilde \omega(\eta) + \tilde \omega(\eta) \otimes \eta, \quad \eta \in \R^4.
$$
By taking the inverse Fourier transform,
we see that $(1-P)\chi f$ is of the form $c g + d^s \omega$
for some $c \in \D'(\R^4)$ and $\omega \in \D'(\R^4; \L1)$.
Let $\chi_2 \in C_0^\infty(M)$
satisfy $\chi_2(x) = 1$.
Now $(x,\xi) \notin \WF(P \chi f)$ implies $(x,\xi) \notin \WF(f - \chi_2(c g + d^s \omega))$, since
\begin{align}
\label{f_is_Pf_plus_h}
\chi f = P \chi f + (1-P)\chi f = P \chi f + c g + d^s \omega.
\end{align}

Suppose now that there are $c \in \E'(M)$ and $\omega \in \E'(M; \L1)$
such that $(x,\xi) \notin \WF(f + c g + d^s \omega)$.
Then $(x,\xi) \notin \WF(P f)$ since
$$
\chi P(f + c g + d^s \omega) = \chi P f,
$$
and the pseudodifferential operator $P$ does not move singularities.
Furthermore, the formula (\ref{f_is_Pf_plus_h})
implies that $(x,\xi) \in \WF(A P \chi f)$
if and only if $(x,\xi) \in \WF(A \chi f)$.
\end{proof}

\section{Microlocal analysis of the light ray transform}
\label{sec_muloc_L}

\newcommand{\pair}[1]{\left\langle #1 \right\rangle}
\def\CC{\mathcal C}
\def\MM{\R^4}

Let us now study the light ray transform in the limited angle case,
that is, when $\UU_1 \subset \R^3$ is an arbitrary open set in Theorem \ref{th_main}. 
Let $\chi \in C_0^\infty(\UU_1)$.
The restricted light ray transform $\chi L$
belongs to the class of operators considered in \cite{Greenleaf1991}.
Here we will complement the results in \cite{Greenleaf1991}
by computing the symbols of the restricted light ray transform 
and its normal operator with a suitable cutoff. 

It is well known that a restricted ray transform is a Fourier integral operator associated with the canonical relation given by the twisted conormal bundle $N^*Z'$ of its point-line relation $Z$ \cite{Guillemin}.
We will next give a parametrization of $N^*Z'$ and the symbol of $\chi L$
in this parametrization. 

We take the Fourier transform with respect to $y \in \R^3$ and get
\begin{align*}
%\label{eq_FAv}
\F L f (\eta,v) = 
&= \int_{\R^3} \int_\R f_{lm}(s, y+sv) \theta^l \theta^m e^{-i\eta y} ds dy	    
\\\nonumber&= \int_{\R^{4}} f_{lm}(x) \theta^l \theta^m e^{-i \eta (x'-x^0 v)} dx.
\end{align*}
Here $\theta = (1,v)$ and we have used the change of variables (\ref{change_of_var_sy}). 
We see that the Schwartz kernel of $\chi L$ 
is the oscillatory integral 
\begin{align}
\label{def_KA}
\int_{\R^3} e^{i \phi(x,y,v,\eta)} 
\aa(x,y,v,\eta) d\eta,
\quad x \in \R^4,\ (y,v) \in \R^3 \times S^2,
\end{align}
where $\aa(x,y,v,\eta) = (2\pi)^{-3} \theta^l \theta^m$ and
$\phi(x,y,v,\eta) = \eta(y-x') + x^0 \eta v$.
We write $\CC = \R^3 \times S^2$.
%and consider the Fourier integral operators $A$ having 
%Schwartz kernel $K_A \in \D'(\R^4 \times \CC)$ of the form (\ref{def_KA})
%with $\aa \in S^m((\R^4 \times \CC) \times \R^3)$, $m \in \R$.
The critical set of $\phi$ is
\begin{align*}
C_\phi &= \{(y,v; x; \eta) \in \CC \times \MM \times (\R^3 \setminus 0);\ 
\phi_\eta' = 0\}
\\&= \{(y,v,x,\eta) ;\ 
y = x' - x^0 v\}.
\end{align*}
The phase function $\phi$ parametrizes the conormal bundle $N^* Z$
of the 
point-line relation
$$
Z = \{ (x' - x^0 v, v; x^0, x') \in \CC \times \MM;\
v \in S^2,\ x \in \R^4 \}
$$
via the diffeomorphism
$C_\phi \ni (y,v,x,\eta) \mapsto (y,v,\phi_{y,v}', x, \phi_x') \in N^* Z$.
Note that the twisted conormal bundle 
\begin{align}\nonumber
N^* Z' = 
\{&(y, v, \eta, w;\, x^0, x', \xi_0, \xi')
\in (T^*\CC \times T^*\MM) \setminus 0; 
\\&\qquad \nonumber
y = x' - x^0v,\ 
\eta = \xi',\ w = x^0 \xi'|_{T_v S^2},\ \xi_0 = -\xi'v,\ 
\\&\qquad \label{Z}
v \in S^2,\ \eta \in \R^3,\ x \in \R^4 \}
\end{align}
is a canonical relation.
Indeed, 
$
N^* Z' \subset (T^* \CC \setminus 0) \times (T^* \MM \setminus 0)$,
since $\eta = 0$ if and only if $\xi' = 0$,
and $\xi' = 0$ implies $\xi_0 = -\xi'v = 0$. 

In the parametrization (\ref{def_KA}), the symbol of $\chi L$  is
$$
(2\pi)^{3/4} \chi(y) \theta^l \theta^m, \quad \theta = (1,v),\ 
(y,v) \in \CC.
$$
The symbol is of order zero 
(\ref{def_KA}) and hence $\chi L$ is a Fourier integral operator of order $-3/4$ associated with the canonical relation $N^* Z'$, see \cite[Def. 3.2.2]{FIO1}.
The theory of Fourier integral operators implies that 
$\chi L$ gives a continuous map from $\mathcal E'(\MM)$
to $\D'(\CC)$. 
Taking into account the conformal invariance, we have proven 
Lemma \ref{L_is_FIO}.

Let us next verify the claims in the end of Section \ref{sec_results} about the sharpness of Theorem \ref{th_main}. 

\begin{lemma}\label{lemma10.1}
\label{lem_lightlike_collinear}
Let $(y, v, \eta, w;\, x, \xi_0, \xi') \in N^* Z'$. Then
$\xi$ is lightlike if and only if $\xi'|_{T_v S^2} = 0$.
%2. $|\xi'v| = |\xi'|$.
Moreover, if $\xi$ is lightlike then $\xi' = -\xi_0 v$.
\end{lemma}
\begin{proof}
We have
$|\xi_0| = |\xi'v| \le |\xi'| |v| = |\xi'|$
with the equality holding if and only if $\xi'$ and $v$ are collinear.
But $\xi'$ and $v$ being collinear is equivalent with $\xi'|_{T_v S^2} = 0$.
Thus $\xi'|_{T_v S^2} = 0$ if and only if $|\xi_0| = |\xi'|$.
Suppose now that $\xi$ is lightlike. Then $\xi = c v$ for some $c \in \R$,
and thus $\xi_0 = -\xi'v = -c$.
\end{proof}

\begin{lemma}
\label{lem_WFA}
Let $(x,\xi) \in T^* \MM \setminus 0$.
Then the relation $C = N^*Z'$ 
mapping $T^* \MM$ to $T^* \CC$ satisfies the following:
\begin{enumerate}
\item If $\xi$ is timelike, then $C(x,\xi) = \emptyset$.
\item If $\xi$ is spacelike, then 
$
C(x,\xi) = \{ (x'-x^0 v , v, \xi', x^0 \xi'|_{T_v S^2});\ v \in S_\xi^1 \},
$ 
where $S_\xi^1$ is the circle defined in Lemma \ref{lem_Fourier}.
\item If $\xi$ is lightlike, then 
$
C(x,\xi) = (x' + \frac {x^0}{\xi_0} \xi', - \frac{\xi'}{\xi_0}, \xi', 0).
$
\end{enumerate}
\end{lemma}
\begin{proof}
If $(y,v;\eta,w) \in T^* \CC$ is in relation $N^* Z'$ with $(x,\xi) \in T^*\MM$, that is, $(y,v;\eta,w) \in N^* Z'(x,\xi)$,
then 
$$|\xi_0| = |\xi'v| \le |\xi|.$$ 
Thus $\xi$ is lightlike or spacelike.
Let us suppose that $\xi$ be spacelike.
We have $y = x'-x^0 v$, $\eta = \xi'$
and $w = x^0 \xi'|_{T_v S^2}$. 
Finally, $\xi_0 = -\xi'v$ is equivalent with $v \in S_\xi^1$.
Let us now suppose that $\xi$ in lightlike.
Note that $\xi_0 \ne 0$ since otherwise $\xi = 0$.
Lemma \ref{lem_lightlike_collinear} implies that $v = -\xi'/\xi_0$
and that $w = 0$.
\end{proof}

The restricted light ray transform $\chi L$ maps singularites through the relation $N^* Z'$, that is,
$\WF(\chi Lf) \subset N^* Z'\ \WF(f)$, see \cite[Th. 4.1.1, Th. 2.5.14]{FIO1}. Hence
Lemma \ref{lem_WFA} implies that $\WF(\chi Lf)$ is empty
if $\WF(f)$ contains only timelike covectors.
Let us now suppose that $(x,\xi) \in T^* \MM$ is spacelike and that 
the circle 
\begin{align}
\label{def_S1xxi}
S_{x,\xi}^1 = \{ x' - x^0 v;\ v \in S_\xi^1 \}
\end{align}
does not intersect $\UU_1$. 
Then Lemma \ref{lem_WFA} implies that $\WF(\chi L f)$ is empty
if $\WF(f)$ is contained in a small conical neighborhood of 
of $(x, \xi)$ in $T^*\MM$.

The condition (\ref{visibility_cond})
is equivalent with $S_{x,\xi}^1 \cap \UU_1 \ne \emptyset$.
Namely, 
if a null geodesic intersects $x$
then, up to an affine reparametrization,
it coincides with 
$\gamma(\tau; x,\theta)$ where $\theta = (1,v)$
for some $v \in S^2$.
Now $\xi(\dot \gamma(0; x,\theta)) = \xi_0 + \xi'v$, 
whence $v \in S_\xi^1$ is equivalent with vanishing of $\xi(\dot \gamma(0; x,\theta))$. Finally, $\gamma(\tau; x,\theta)$
%$$\gamma(\tau; x,\theta) = (x^0 + \tau, x' + \tau v)$$
intersects $\{0\} \times \UU_1$ for some $v \in S_\xi^1$ 
if and only if $S_{x,\xi}^1 \cap \UU_1 \ne \emptyset$.

Let us now suppose that $(x,\xi) \in T^* \MM \setminus 0$ is lightlike,
and consider the null geodesic 
\begin{align}
\label{lightlike_nonuniqueness_mu}
\mu(\tau) = (x^0 + \tau, x' - \tau \xi_0^{-1} \xi'), \quad \tau \in \R.
\end{align}
Then $N^* Z'$ maps all the points $(\mu(\tau), \xi) \in T^* \MM$, $\tau \in \R$,
on the same point in $T^* \CC$.
We expect that this leads to artifacts when trying to recover lighlike singularities, however, we do not analyze this further in the present paper.

\section{Microlocal analysis of the normal operator}
\label{sec_microlocal_normalop}

Let us consider an operator $A$ having 
kernel of the form (\ref{def_KA}),
and use the following notation for the projections
\begin{align}
\label{NZ_projections}
\xymatrix{
& N^*Z' \ar[ld]_\pi \ar[dr]^\rho &   
\\
T^*\MM & & T^*\CC }
\end{align}   
When the projection $\rho$ is an injective immersion, 
the diagram (\ref{NZ_projections}) is said to satisfy the Bolker condition.
In this case the corresponding normal operator $A^* A$ is a pseudodifferential operator \cite{Guillemin}.
In our case, the Bolker condition holds only outside the set 
\def\LL{\mathscr L}
$$
\LL = \{(y, v, \eta, w;\, x, \xi) \in T^*\CC \times T^*\MM;\ \text{$\xi$ is lightlike}\}.
$$

\begin{lemma}
\label{lem_rho_lightlike}
The restriction of $\rho$ on $N^* Z' \setminus \LL$ is an injective immersion
but $d\rho$ fails to be injective on $N^* Z'  \cap \LL$.
\end{lemma}
\begin{proof}
We may use $(y,v,\xi',x^0) \in \CC \times \R^{3+1}$
as coordinates on $N^* Z$.
In these coordinates
\begin{align*}
\rho = (y, v, \xi', x^0 \xi'|_{T_v S^2}),
\quad
d \rho = \begin{pmatrix}
id &0&0&0\\
0& id &0&0\\
0&0& id &0\\
0& * & * & \xi'|_{T_v S^2}
\end{pmatrix},
\end{align*}
where $*$ denotes an element that does not play a role in the proof.
We see that  $d\rho$ is injective if and only if $\xi'|_{T_v S^2} \ne 0$.
By Lemma \ref{lem_lightlike_collinear} this is equivalent with $\xi$
being lightlike.

Let us now show global injectivity of the restriction. 
%We identify $N^* Z$ with the parameter space $\CC \times \R^{3+1}$.
Let $(y,v,\xi',x^0)$ and $(\tilde y, \tilde v, \tilde \xi', \tilde x^0)$ be in 
$\CC \times \R^{3+1}$, and suppose that $\xi$ is not lightlike and that $\rho$ maps the corresponding points in $N^* Z$ on the same point.
Then $\tilde y = y$, $\tilde v = v$, $\tilde \xi' = \xi'$, and 
$\tilde x^0 \xi'|_{T_v S^2} = x^0 \xi'|_{T_v S^2}$.
Moreover, $\tilde x^0 = x^0$ since $\xi'|_{T_v S^2} \ne 0$ as $\xi$ is not lightlike. 
\end{proof}

We will show next that $A^* A$ is 
a pseudodifferential operator if we microlocalize away from the set $\LL$.

\begin{lemma}
\label{lem_WF_composition}
Let $(x,\xi) \in T^* \MM \setminus 0$.
Then the composed relation $C = (N^*Z')^{-1} \circ N^*Z'$ 
mapping $T^* \MM$ to itself satisfies
\begin{enumerate}
\item If $\xi$ is timelike, then $C(x,\xi) = \emptyset$.
\item If $\xi$ is spacelike, then $C(x,\xi) = (x,\xi)$.
\item If $\xi$ is lightlike, then 
$C(x,\xi) = 
\{ (\mu(\tau); \xi);\ \tau \in \R \}$,
where $\mu$ is the null geodesic (\ref{lightlike_nonuniqueness_mu}).
\end{enumerate}
\end{lemma}
\begin{proof}
The timelike case follows immediately from Lemma \ref{lem_WFA}.
Let $\xi$ be spacelike, and let $(y,v;\eta,w) \in T^* \CC$ be in relation $N^* Z'$ with $(x,\xi)$
and with a point $(\tilde x, \tilde \xi)$.
We have $\tilde \xi' = \eta = \xi'$. This implies that also $\tilde \xi^0 = -\xi'v = \xi^0$.
As $\xi'|_{T_v S^2} \ne 0$ by Lemma 
\ref{lem_lightlike_collinear}, the equation $\tilde x^0 \xi'|_{T_v S^2} = w = x^0 \xi'|_{T_v S^2}$ implies 
$\tilde x^0 = x^0$. Finally $\tilde x' - x^0 v = y = x' - x^0 v$,
whence $\tilde x' = x'$.

Let $\xi$ be lightlike.
Then Lemma \ref{lem_WFA} says that
$
(x' + \frac {x^0}{\xi_0} \xi', - \frac{\xi'}{\xi_0}, \xi', 0)
$
is the only point in $T^* \CC$ that is in relation $N^* Z'$ with $(x,\xi)$.
This determines $\xi$ uniquely but $x$ only up to a translation along $\mu$.
\end{proof}

\begin{lemma}
\label{lem_clean_intersection}
Let us define $C = N^*Z' \setminus \LL$.
The composition $C^{-1} \circ C$
is clean in the sense of \cite[Th. 21.2.14]{HormanderVol4}. %, and it is given by $\diag\pi(C)$.
%Moreover, $\pi(C) = \{(x, \xi) \in T^* M;\ \text{$\xi$ is spacelike}\}$, 
The projection $\pi_2$ from the intersection of
$C^{-1} \times C$
with $T^* \MM \times \diag(T^* \CC) \times T^* \MM$ to $(T^*\MM)^2$
is proper, and the fibers $\pi_2^{-1}(p)$, $p \in (T^*\MM)^2$,
are connected.
\end{lemma}
\begin{proof}
We write
$$
X = C^{-1} \times C,
\quad 
Y = T^* \MM \times \diag(T^* \CC) \times T^* \MM.
$$
We need to show that $X \cap Y$ is a smooth manifold and that 
\begin{align}
\label{clean_intersection_Tp}
T_p (X \cap Y) = T_p X \cap T_p Y, \quad p \in X \cap Y.
\end{align}
Lemma \ref{lem_WF_composition} implies that $X \cap Y$ coincides with the smooth manifold $\diag(C)$ after a reordering of the factors in the Cartesian product.
We may use $p = (y,v,\xi',x^0) \in \CC \times \R^{3+1}$
as coordinates on $X \cap Y$.
%and $(p, q) \in (\CC \times \R^{3+1})^2$ as coordinates on $X$.
Notice that a tangent vector of $X$, parametrized by 
$(\delta p, \delta q) \in T_{(p,p)} (\CC \times \R^{3+1})^2$,
coincides with a tangent vector of $Y$ if and only if 
$d\rho\, \delta p = d\rho\, \delta q$.
This implies $\delta p = \delta q$ since $d\rho$ is injective on $C$.
Thus (\ref{clean_intersection_Tp}) holds.

%\def\TT{\mathcal T}
%We have $N^* Z \cap \TT = \emptyset$ where 
%$$
%\TT = \{(y, v, \eta, w;\, x, \xi) \in T^*\CC \times T^*M;\ \text{$\xi$ is timelike}\}.
%$$
%Indeed, %if $(y,v,\eta,w; x^0, x', \xi^0, \xi') \in N^* Z$ then

A fiber $\pi_2^{-1}(p)$, $p \in (T^*\MM)^2$,
is diffeomorphic with the circle $S_\xi^1$ if $p = (x,\xi,x,\xi)$ and $\xi$ is spacelike,
and it is empty otherwise. 
%The equation $\xi^0 = -\xi'(v)$ implies that 
%$v \in S^2$ is in 
%the affine plane $P_\xi$ defined by (\ref{def_affplane_P}).
%If $\xi$ is timelike then $P_\xi \cap S^2 = S_\xi^1 = \emptyset$.
%If $(x,\xi) \in T^* M$ and $\xi$ is spacelike, 
%then $S_\xi^1$ is diffeomorphic to the circle $S^1$, and 
%also the fiber $\pi^{-1}(x,\xi)$ is diffeomorphic to $S^1$.
%Furthermore, 
\end{proof}

Lemma \ref{lem_clean_intersection} says that 
if we cut away the lightlike covectors, then 
we can use the clean intersection calculus \cite[Th. 25.2.3]{HormanderVol4}.
Let $\chi_0 \in S^0(T^* \MM)$ 
satisfy $\supp(\chi_0) \cap \LL_\pi = \emptyset$ where
$$\LL_\pi = \{(x,\xi) \in T^* \MM;\ \text{$\xi$ is lightlike} \}.$$
Then $A^* A \chi_0$ is a Fourier integral operator with 
the canonical relation $C^{-1} \circ C$ where $C = N^*Z' \setminus \LL$,
and Lemma \ref{lem_WF_composition} implies
that $C^{-1} \circ C$ is in $N(\diag(\MM))'$.
Hence $A^* A \chi_0$ is a pseudodifferential operator. 
The calculus implies also that the order of $A^* A \chi_0$
is $-1$ if the order of $A$ is $-3/4$ 
since the fibers $\pi_2^{-1}(p)$, $p \in (T^*\MM)^2$,
are one dimensional (i.e. the excess of the composition is one). 

If we are given the data $Af$ %on $\UU_1 \times S^2$ 
we can not compute $A \chi_0 f$,
but we can choose pseudodifferential operators 
$\chi_0$ on $\MM$ and $\chi_1$ on $\CC$
so that $\chi_1 A = \chi_1 A \chi_0$ modulo a smoothing operator. 
Altenatively, we could choose pseudodifferential operators $\chi_1$
and $\chi_0$, both
on $\MM$, so that $\chi_1 A^* A \chi_0 = \chi_1 A^* A$
modulo a smoothing operator. 
By Lemma \ref{lem_WF_composition} 
this follows if 
 $\chi_0 = 1$ in $\supp(\chi_1)$
and $\supp(\chi_0) \cap \LL_\pi = \emptyset$.
However, we prefer using $\chi_1$ on $\CC$
since we need also a cut off $\chi_2 \in C_0^\infty(\UU_1)$
due to the fact that we do not have data on $\R^3 \setminus \UU_1$.

\begin{proposition}
\label{prop_principal_symbol}
Let $\tilde  \chi_1 \in C_0^\infty(\R)$
satisfy $\supp(\tilde \chi_1) \subset (-1,1)$, and define
$$
\chi_1(v,\eta) = \tilde \chi_1(|\eta v|/|\eta|),
\quad v \in S^2,\ \eta \in \R^3.
$$
Let $\UU_1 \subset \R^3$ be open and let $\chi_2 \in C_0^\infty(\UU_1)$.
Then the principal symbol of the pseudodifferential operator 
$L^* \chi_1 \chi_2 L$ is
$$
\sigma_0(x,\xi) = \frac{2\pi \tilde \chi_1(|\xi_0|/|\xi'|)}{(|\xi'|^2 - |\xi_0|^2)^{1/2}} 
\int_{S_\xi^1} \chi_2(x'-x^0 v) \theta^j \theta^k \theta^l \theta^m dv,
$$
where $\theta = (1,v)$.
\end{proposition}
\begin{proof}
We choose $\tilde \chi_0 \in C_0^\infty(-1,1)$
such that 
$\tilde \chi_0 = 1$ in $\supp(\tilde \chi_1)$, and define
$\chi_0(\xi) = \tilde \chi_0(|\xi_0|/|\xi'|)$.
If $(y, v, \eta, w;\, x, \xi_0, \xi') \in N^* Z'$
then $\eta = \xi'$ and $\xi_0 = - \xi'v = -\eta v$.
Thus $|\eta v|/|\eta| = |\xi_0|/|\xi'|$
and $\chi_0(\xi) = 1$ if $\chi_1(v,\eta) \ne 0$.
Now $\chi_1 L = \chi_1 L \chi_0$ modulo a smoothing operator
and $\supp(\chi_0) \cap \LL_\pi = \emptyset$.
Hence $L^* \chi_1 \chi_2 L$ is a pseudodifferential operator.

The kernel of $L^* \chi_1 \chi_2 L$ is of the form
\begin{align*}
\int_{\CC \times \R^3 \times \R^3} e^{-i \phi(x,y,v,\xi')} e^{i \phi(z,y,v,\eta)} 
\aa(z,v,\eta) d\eta d\xi' dy dv,
\quad x,z \in \R^4,
\end{align*}
where $\aa(z,v,\eta) = (2\pi)^{-3} \chi_1(v, \eta) \chi_2(z'-z^0 v) \theta^j \theta^k \theta^l \theta^m$.
We have 
$$
\phi(z,y,v,\eta)
- \phi(x,y,v,\xi')
= (\eta - \xi')y - \tilde \phi(z,v,\eta) + \tilde \phi(x,v,\xi'),
$$
where $\tilde \phi(x,v,\xi') = \xi'x' - x^0 \xi'v$.
Hence the kernel simplifies to
\begin{align*}
&\int_{S^2 \times \R^3 \times \R^3} \delta(\eta - \xi') 
e^{i \tilde \phi(x,v,\xi') - i \phi(z,v,\eta)} 
\aa(z,v,\eta) d\eta d\xi' dv
\\&\quad=
\int_{S^2 \times \R^3} 
e^{i \xi'(x'-z') - i(x^0 - z^0) \xi'v} 
\aa(z,v,\xi') d\xi' dv
\\&\quad= \int_{\R^4} \int_{S^2}
e^{i \xi (x-z)} \kappa_\xi^* \delta(v) \aa(z,v,\xi') dv d\xi,
\end{align*}
where $\kappa_\xi(v) = \xi_0 + \xi'v$.
By (\ref{kappaxi_pullback}) the symbol $\sigma(z,\xi)$ of 
$L^* \chi_1 \chi_2 L$ is 
\begin{align*}
%\label{the_symbol}
\sigma(z,\xi) 
&= 
(2\pi)^4 (|\xi'|^2 - |\xi_0|^2)^{-1/2} \int_{S_\xi^1} \aa(z,v,\xi') dv
\\\nonumber&=
\frac{2\pi \tilde \chi_1(|\xi_0|/|\xi'|)}{(|\xi'|^2 - |\xi_0|^2)^{1/2}} 
\int_{S_\xi^1} \chi_2(z'-z^0 v) \theta^j \theta^k \theta^l \theta^m dv.
\end{align*}
Thus $\sigma_0(x,\xi) = \sigma(x,\xi)$.
\end{proof}

\begin{proof}[Proof of Theorem \ref{th_main}]
By the conformal invariance it is enough consider the Minkowski case. 
Let $(x,\xi)$ be spacelike and choose the cutoff $\tilde \chi_1$ in 
Proposition \ref{prop_principal_symbol}
so that $\tilde \chi_1 = 1$ near $|\xi_0|/|\xi'|$.
Then $\sigma_0(x,\xi) = c(\xi) N^{jklm}$,
where $c(\xi) \ne 0$ and $N^{jklm}$ is the tensor (\ref{def_N})
with $\chi(v) = \chi_2(x'-x^0 v)$.
We recall that the 
visibility condition (\ref{visibility_cond})
is equivalent to $S_{x,\xi}^1 \cap \UU_1 \ne \emptyset$,
where the circle $S_{x,\xi}^1$ is defined by (\ref{def_S1xxi}).
Thus we can choose non-negative $\chi_2 \in C_0^\infty(\UU_1)$
so that $\chi$ does not vanish identically on $S_\xi^1$. 

Let $P$ be the projection (\ref{the_projection}).
Note that both $N$ and $P$ are homogeneous of degree zero, 
and 
$W_0 = N + (1-P)$
is elliptic near $(x,\xi)$. Indeed, if $W_0 f  = 0$ then 
$$
0 = PW_0 f = NPf, \quad 0 = (1-P)W_0 f = (1-P)f.
$$
Thus $P f \in \ker(N) = \ker(P)$ and $f = Pf = P^2 f = 0$.

As $W_0$ is the principal symbol of 
the zeroth order pseudodifferential operator 
$$
W = c^{-1} L^* \chi_1 \chi_2 L + (1-P),
$$
there is a parametrix $Q$ of $W$ such that $W = 1$ 
modulo a smoothing operator 
near $(x,\xi)$.
Thus $P =  Q W P = Q c^{-1} L^* \chi_1 \chi_2 L P$.
Hence
$$
(x,\xi) \in \WF(Pf) \quad \text{if and only if} 
\quad (x,\xi) \in \WF( L^* \chi_1 \chi_2 L Pf),
$$
and we apply Lemma \ref{lem_microlocal_nullsp} to conclude.
\end{proof}

\medskip

\noindent{\bf Acknowledgements.} 
The authors express their gratitude to
the  Institute Henri Poincar\'e where a part of this work has been done.
ML was partly supported by the Academy of Finland project 272312 and the Finnish Centre of Excellence in Inverse Problems Research 2012-2017.
LO was partly supported by the Engineering and Physical Sciences Research Council, UK, grant EP/L026473/1.
PS was partly supported by an NSF FRG Grant DMS-1301646.
GU was partly supported by NSF, Simons Fellowship, and a FiDi Professorship at the University of Helsinki.

\bibliographystyle{abbrv}
\bibliography{refs}

\end{document}